\newcommand\fontSizeIfFig{1}

\ifdefined\TWOCOLUMN
\newcommand\twoColFigSize{2}
\documentclass[conference]{IEEEtran}

\else
\documentclass[journal,onecolumn,twoside]{IEEEtran}
\newcommand\twoColFigSize{1}
\fi

\usepackage{graphicx}
\usepackage{amssymb,amsmath,amsthm}
\usepackage{subcaption,siunitx,booktabs}
\usepackage[table]{xcolor}
\usepackage{colortbl}
\definecolor{Gray}{gray}{0.85}
\newcolumntype{B}{>{\columncolor{Gray}}c}
\usepackage{psfrag}
\usepackage{cite}
\usepackage{longfbox}
\newcommand\given[1][]{\:#1\vert\:}
\usepackage{float}
\restylefloat{table}

\usepackage{tikz}
\usetikzlibrary{fit,matrix}
\usetikzlibrary{patterns,snakes}

\newcommand{\tikzmark}[2]{\tikz[overlay,remember picture,baseline]
\node [anchor=base] (#1) {$#2$};}

\newcommand{\DrawVLine}[3][]{%
  \begin{tikzpicture}[overlay,remember picture]
    \draw[shorten <=0.3ex, #1] (#2.north) -- (#3.south);
  \end{tikzpicture}
}

\newtheorem{lemma}{Lemma}
\newtheorem{corollary}{Corollary}
\newtheorem{theorem}{Theorem}
\newtheorem{definition}{Definition}
\newtheorem{proposition}{Proposition}

\newtheorem{remark}{Remark}

\DeclareMathAlphabet{\mathbsf}{OT1}{cmss}{bx}{n}
\DeclareMathAlphabet{\mathssf}{OT1}{cmss}{m}{sl}
\DeclareMathAlphabet{\mathcsf}{OT1}{cmss}{sbc}{n}

\newcommand{\svv}[1]{\mathbf{#1}}

\DeclareSymbolFont{bsfletters}{OT1}{cmss}{bx}{n}  
\DeclareSymbolFont{ssfletters}{OT1}{cmss}{m}{n}
\DeclareMathSymbol{\bsfGamma}{0}{bsfletters}{'000}
\DeclareMathSymbol{\ssfGamma}{0}{ssfletters}{'000}
\DeclareMathSymbol{\bsfDelta}{0}{bsfletters}{'001}
\DeclareMathSymbol{\ssfDelta}{0}{ssfletters}{'001}
\DeclareMathSymbol{\bsfTheta}{0}{bsfletters}{'002}
\DeclareMathSymbol{\ssfTheta}{0}{ssfletters}{'002}
\DeclareMathSymbol{\bsfLambda}{0}{bsfletters}{'003}
\DeclareMathSymbol{\ssfLambda}{0}{ssfletters}{'003}
\DeclareMathSymbol{\bsfXi}{0}{bsfletters}{'004}
\DeclareMathSymbol{\ssfXi}{0}{ssfletters}{'004}
\DeclareMathSymbol{\bsfPi}{0}{bsfletters}{'005}
\DeclareMathSymbol{\ssfPi}{0}{ssfletters}{'005}
\DeclareMathSymbol{\bsfSigma}{0}{bsfletters}{'006}
\DeclareMathSymbol{\ssfSigma}{0}{ssfletters}{'006}
\DeclareMathSymbol{\bsfUpsilon}{0}{bsfletters}{'007}
\DeclareMathSymbol{\ssfUpsilon}{0}{ssfletters}{'007}
\DeclareMathSymbol{\bsfPhi}{0}{bsfletters}{'010}
\DeclareMathSymbol{\ssfPhi}{0}{ssfletters}{'010}
\DeclareMathSymbol{\bsfPsi}{0}{bsfletters}{'011}
\DeclareMathSymbol{\ssfPsi}{0}{ssfletters}{'011}
\DeclareMathSymbol{\bsfOmega}{0}{bsfletters}{'012}
\DeclareMathSymbol{\ssfOmega}{0}{ssfletters}{'012}

\DeclareRobustCommand{\prob}[1][{\rm Pr}]{\ensuremath {{#1}}}

\begin{document}

\allowdisplaybreaks
%

\title{Streaming Erasure Codes over Multi-hop Relay Network}
\ifdefined\TWOCOLUMN
\author{
\IEEEauthorblockN{Elad Domanovitz and Ashish Khisti}
\IEEEauthorblockA{
Department of Electrical and Computer Engineering\\
  University of Toronto\\
                    Toronto, ON M5S 3G4, Canada\\
                    E-mail: \texttt{\{elad.domanovitz, akhisti\}@utoronto.ca}
}  \and
\IEEEauthorblockN{Wai-Tian Tan, Xiaoqing Zhu, and John Apostolopoulos}
\IEEEauthorblockA{
Enterprise Networking Innovation Labs, Cisco Systems\\
Cisco Systems\\
San Jose, CA 95134, USA
}
}
\else
\author{Elad Domanovitz, Ashish Khisti, Wai-Tian Tan, Xiaoqing Zhu, and John Apostolopoulos 
\thanks{E. Domanovitz and A. Khisti are with the Department of Electrical and Computer Engineering, University of Toronto, Toronto, ON M5S 3G4, Canada (email: elad.domanovitz@utoronto.ca, akhisti@ece.utoronto.ca)}
\thanks{W.-T. Tan, X. Zhu, and J. Apostolopoulos are with Cisco Systems, San Jose, CA 95134, USA.}
\thanks{The material in this paper was presented in part at the 2020 IEEE
International Symposium on Information Theory, Los Angeles, CAL.}
}
\fi
\maketitle

\begin{abstract}
This paper studies low-latency streaming codes for the multi-hop network. The source is transmitting a sequence of messages (streaming messages) to a destination through a chain of relays where each hop is subject to packet erasures. Every source message has to be recovered perfectly at the destination within a delay constraint of $T$ time slots. In any sliding window of $T+1$ time slots, we assume no more than $N_j$ erasures introduced by the $j$'th hop channel. The capacity in case of a single relay (a three-node network) was derived by Fong \cite{fong2018optimal}, et al. While the converse derived for the three-node case can be extended to any number of nodes using a similar technique (analyzing the case where erasures on other links are consecutive), we demonstrate next that the achievable scheme, which suggested a clever symbol-wise decode and forward strategy, can not be straightforwardly extended without a loss in performance. The coding scheme for the three-node network, which was shown to achieve the upper bound, was ``state-independent'' (i.e., it does not depend on specific erasure pattern). While this is a very desirable property, in this paper, we suggest a ``state-dependent'' (i.e., a scheme which depends on specific erasure pattern) and show that it achieves the upper bound up to the size of an additional header. Since, as we show, the size of the header does not depend on the field size, the gap between the achievable rate and the upper bound decreases as the field size increases.
\end{abstract}

\section{Introduction}

Real-time interactive video streaming is an integral part of the day-to-day activity of many people in the world. Traditionally, most of the traffic on the internet is not sensitive to the typical delay induced by the network. However, as networks evolved, more and more people are using the network for real-time conversations, video conferencing, and on-line monitoring. According to \cite{cisco2018cisco}, IP video traffic will account for 82 percent of traffic by 2022. Further, live video is projected to grow 15-fold to reach 17 percent of Internet video traffic by 2022.

All types of traffic are susceptible to errors, and therefore many applications use an error-correcting mechanism. One fundamental difference between real-time video streaming and other types of traffic is the (much more stringent) latency requirement each packet has to meet in order to provide a good user experience. A very common error-correcting mechanism is automatic repeat request (ARQ). Using ARQ means that the latency (in case of an error) is at least three times the one-way delay, which in many cases may violate the latency requirements for real-time interactive video streaming.

An alternative method for handling errors in the transmission is forward error correction (FEC). Using FEC has the potential to lower the recovery latency since it does not require communication between the receiver and transmitter. However, in many cases, when FEC is designed, the emphasis is on its error-correcting capabilities while ignoring latency constraints. Two commonly used codes are Low-density parity-check (LDPC) \cite{gallager1962low,mackay1996near} and digital fountain codes \cite{luby2002lt,shokrollahi2006raptor}. The typical block length of these codes is very long (usually a few hundreds of symbols) hence precluding their use for real-time interactive applications.

Low-latency FEC codes are already implemented and have a noticeable impact on the quality of real-time interactive applications. Typically, maximum-distance separable (MDS) codes are used to transmit an extra parity-check packet per every two to five packets \cite{wang2010chitchat}. For example, in \cite{huang2010could}, the FEC implemented in Skype is described, and it is argued that this mechanism is one of the main contributors to the success of this application. 

Memory Maximum Distance Separable convolutional codes (m-MDS) discussed in \cite{justesen1974maximum,gabidulin1988convolutional,gluesing2006strongly} are a class of codes that guarantee decoding assuming the decoder has received sufficiently many parity-check packets. These codes serve as the baseline codes for achieving point-to-point capacity in channels with arbitrary erasures. While in traditional (systematic) m-MDS codes, the parity is appended to the data, in \cite{Karzand2017}, a design of a FEC aimed to reduce the end-to-end average in-order delay (i.e., the time packets spend at the receiver before they can be further processed which was studied in \cite{joshi2012playback}) was described in which a packet composed only from parity symbols is transmitted at a pre-defined rate. While this work considered inserting a single parity packet in each interval, in \cite{LiZhang2020}, this concept was extended to allow insertion of multiple packets in each interval, and thus, the benefit of this FEC was extended to highly lossy links. 

In another line of work, \cite{MartinianSundberg2004} derived the capacity of low-latency FEC (while denoting it as streaming codes) for a (deterministic) channel with bursts of erasures. This work was followed by a plurality of works \cite{badr2013streaming,badr2017layered,fong2019optimal,krishnan2018rate,domanovitz2019explicit,leong2012erasure,KrishnanLowField2020} which extended the channel model to contain both bursts and arbitrary erasures while analyzing a sliding window model. 


While all the works mentioned above-analyzed streaming codes for point-to-point channels, in \cite{fong2018optimal}, the performance of streaming codes for the three-node network was analyzed. As a first step, the channel model used in this paper can be denoted as ``deterministic arbitrary erasure channel''. In this channel model, the location of erasures can be arbitrary (does not necessarily occur in bursts). However, the number of erasures is (globally) upper bounded. For this model, the capacity of streaming codes was first established. Then, it was shown that the derived results also hold for a sliding window model in which only the number of erasures in any given window is upper bounded (i.e., the limit on the global number of erasures was removed).

When analyzing achievable schemes, in \cite{fong2018optimal}, a straightforward extension of point-to-point codes to a setup with a relay was first described. In this extension, each hop uses a point-to-point code, and optimization is carried on the allocation of the delay to be utilized by each code. Denoting this approach as ``message-wise'' decode and forward strategy (since each message is fully decoded at the relay prior to forwarding it), \cite{fong2018optimal} showed that this scheme is inferior to a more sophisticated scheme denoted as ``symbol-wise'' decode and forward in which the relay forwards the recovered symbols (before the entire message can be decoded) to the destination. Further, showing that the rate achieved by ``symbol-wise'' decode and forward coincides with the upper bound resulted in the capacity of the three-node relay network.

Analyzing the capacity of the three-node network shows that when constraints are imposed per segment rather than globally (while meeting the same global requirements), the capacity increases. For example, as we demonstrate next, treating the network as a single-hop link with a maximum of $N=N_1+N_2$ erasures and a total delay constraint of $T$ symbols is worse than analyzing a three-node network where a maximum of $N_1$ erasures are expected in the first segment and a maximum of $N_2$ erasures are expected in the second segment with a total delay of $T$ symbols. However, internet paths almost never consist of only a single relay (see, e.g., \cite{begtasevic2001measurements,mukaddam2011hop}). Hence, designing a streaming code for a path consisting of multiple links when possible (i.e., take into account the error behavior of each link rather than aggregate across all of the links) is expected to result in improved performance guarantees.

The coding scheme described in \cite{fong2018optimal}, which was shown to achieve the capacity of the network, has another appealing property, which is ``state-independent'', i.e., it does not depend on the specific location of erasures in the different segments. Unfortunately, as we show next, there is no straightforward extension of this scheme to a more general case (a network with more than three nodes). The scheme suggested in this paper for any number of relays is a ``state-dependent'' scheme, i.e., it is a scheme which reorders the symbols transmitted by each relay based on the erasure patterns that occurred in the previous links. While requiring an additional header to allow each relay to encode the received symbols transmitted to the next relay, we show that it can be easily used for any number of relays.


In this paper, we first extend the upper bound derived in \cite{fong2018optimal} to the general case of a multi-hop relay network. We then describe the state-dependent scheme for the general $L$ relay scenario and show it achieves the upper bound up to an additional overhead (a required header). We further show that the size of the header is a function of the required delay and the erasure pattern (hence it does not depend on the field size used by the code). Therefore, the gap from the upper bound decreases as the field size increases.

The rest of this paper is organized as follows. Section~\ref{sec:networkModel} outlines the network model of interest. Section~\ref{sec:standardDefAndKnownRes} presents the formulation of streaming codes and outlines the known results for basic network models. In this Section, the problem is defined as coding over the deterministic erasure model. Section~\ref{sec:mainRes} presents the main results of this paper. Section~\ref{sec:motExample} provides a motivating example. In this example, we show that a straightforward extension of the (state-independent) achievable scheme for a single relay results in a loss in the minimum delay that can be achieved (while maintaining the same rate) compared to the state-dependent scheme. Section~\ref{sec:converse} contains the proof for the upper bound on the achievable rate for a network that consists of any number of relays. Section~\ref{sec:codScheme} presents the state-dependent symbol-wise decode and forward coding scheme and contains proof on its achievable rate. Section~\ref{sec:upperBound} provides an upper bound on the error probability when using the state-dependent symbol-wise decode and forward coding scheme when used over a channel with random (i.i.d.) erasures. Section~\ref{sec:numerical} provides numerical results for different coding schemes used over four-node (two relay) network with random (i.i.d.) erasures. Finally, Section~\ref{sec:extToSliding} provides an extension of the presented results to the sliding window channel.

\begin{figure*}
    \begin{center}
        \begin{psfrags}
            \psfrag{s}[][][0.75]{$s=r_0$}
            \psfrag{d}[][][0.75]{$d=r_{L+1}$}
            \psfrag{r_1}[][][1]{$r_1$}
            \psfrag{r_2}[][][1]{$r_2$}
            \psfrag{r_k}[][][1]{$r_L$}
            \psfrag{N1}[][][1]{$N_1~{\rm Erasures}$}
            \psfrag{N2}[][][1]{$N_2~{\rm Erasures}$}
            \psfrag{Nk}[][][1]{$N_{L+1}~{\rm Erasures}$}
            \psfrag{A}[][][0.75]{${\bf s}_i\in\mathbb{F}^k$}
            \psfrag{B}[][][0.75]{${\bf x}_i^{(r_0)}\in\mathbb{F}^{n_1}$}
            \psfrag{C}[][][0.75]{${\bf y}_i^{(r_1)}\in\mathbb{F}^{n_1}$}
            \psfrag{D}[][][0.75]{${\bf x}_i^{(r_1)}\in\mathbb{F}^{n_2}$}
            \psfrag{E}[][][0.75]{${\bf y}_i^{(r_2)}\in\mathbb{F}^{n_2}$}
            \psfrag{F}[][][0.75]{${\bf x}_i^{(r_2)}\in\mathbb{F}^{n_3}$}
            \psfrag{G}[][][0.75]{${\bf y}_i^{(r_L)}\in\mathbb{F}^{n_{L}}$}
            \psfrag{H}[][][0.75]{${\bf x}_i^{(r_L)}\in\mathbb{F}^{n_{L+1}}$}
            \psfrag{I}[][][0.75]{${\bf y}_i^{(r_{L+1})}\in\mathbb{F}^{n_{L+1}}$}
            \psfrag{J}[][][0.75]{$\hat{\bf s}_{i-T}\in\mathbb{F}^{k}$}
            \includegraphics[width=\twoColFigSize\columnwidth]{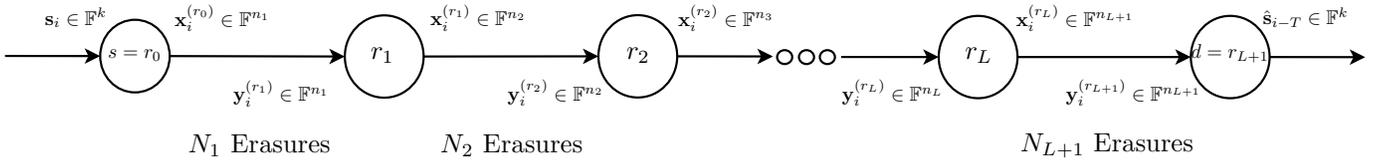}
        \end{psfrags}
    \end{center}
    \caption{Symbols generated in the $L$-node relay network at time $i$.}
    \label{fig:fig1}
\end{figure*}

\subsection{Network Model}
\label{sec:networkModel}
A source node wants to send a sequence of messages $\{{\bf s}_i\}_{i=0}^\infty$
to a destination node with the help of $L$ middle nodes $r_1,\ldots,r_L$. To ease notation we denote the source node as $r_0$, and destination node as $r_{L+1}$. Let $k$ be a non-negative integer, and $n_1,n_2,\ldots,n_{L+1}$ be $L+1$ natural numbers.

Each ${\bf s}_i$ is an element in $\mathbb{F}^k$ where $\mathbb{F}$ is some finite field. In each time slot $i\in\mathbb{Z}_+$, the source message ${\bf s}_i$ is encoded into a length-$n_1$ packet ${\bf x}_i^{(r_0)}\in\mathbb{F}^{n_1}$ to be transmitted to the first relay through the erasure channel $(r_0,r_1)$. The relay receives ${\bf y}^{(r_1)}_i\in\mathbb{F}^{n_1}\cup\{*\}$ where ${\bf y}^{(r_1)}_i$ equals either ${\bf x}_i^{(r_0)}$ or the erasure symbol $``*"$. In the same time slot, relay $r_1$ transmits ${\bf x}_i^{(r_1)}\in\mathbb{F}^{n_2}$ to relay $r_2$ through the erasure channel $(r_1, r_2)$. Relay $r_2$ receives ${\bf y}^{(r_2)}_i\in\mathbb{F}^{n_2}\cup\{*\}$ where ${\bf y}^{(r_2)}_i$ equals either ${\bf x}_i^{(r_1)}$ or the erasure symbol $``*"$. The same process continues (in the same time slot) until relay $r_L$ transmits ${\bf x}_i^{(r_L)}\in\mathbb{F}^{n_{L+1}}$ to the destination $r_{L+1}$ through the erasure channel $(r_L, r_{L+1})$. To simplify the analysis we note that we assume zero propagation delay and zero processing delay for the transmission. Hence, in case no coding is applied ($n_1=n_2=\ldots=n_{L+1}=k$) and no erasures occur, ${\bf y}^{(r_{L+1})}_i=s_i$. When such assumptions are relaxed, extensions to the results described in the paper can be naturally described (see, e.g. \cite{badr2017perfecting}).



We first assume that on the discrete timeline, each channel $(r_{j-1}, r_{j})$ introduces up to $N_j$ arbitrary erasures respectively. The symbols generated in the $L$-node relay network at time $i$ are illustrated in Figure~\ref{fig:fig1}.


\subsection{Standard Definitions and Known Results}
\label{sec:standardDefAndKnownRes}
\begin{definition}
An $(n_1,n_2,\ldots,n_{L+1},k,T)_{\mathbb{F}}$-streaming code consists of the following:
\begin{enumerate}
    \item A sequence of source messages $\{{\bf s}_i\}_{i=0}^{\infty}$ where ${\bf s}_i \in \mathbb{F}^k$.
    \item An encoding function $f_i^{(r_0)}:\underbrace{\mathbb{F}^k\times\ldots\times\mathbb{F}^k}_{i+1~{\rm times}}\to\mathbb{F}^{n_1}$ for each $i\in\mathbb{Z}_{+}$, where $f_i^{(r_0)}$ is used by node $r_0$ at time $i$ to encode ${\bf s}_i$ according to
    \begin{align*}
        {\bf x}_i^{(r_0)}=f_i^{(r_0)}\left({\bf s}_0,{\bf s}_1,\ldots,{\bf s}_i\right).
    \end{align*}
    \item A relaying function for node $r_j$ where $j\in\{1,\ldots,L\}$, \mbox{$f_i^{(r_j)}:\underbrace{\mathbb{F}^{n_{j}}\cup\{*\}\times\ldots\times\mathbb{F}^{n_j}\cup\{*\}}_{i+1~{\rm times}}\to\mathbb{F}^{n_{j+1}}$} for each $i\in\mathbb{Z}_{+}$, where $f_i^{(r_j)}$ is used by node $r_j$ at time $i$ to construct
    \begin{align*}
        {\bf x}_i^{(r_j)}=f_i^{(r_j)}\left({\bf y}_0^{(r_j)},{\bf y}_1^{(r_j)},\ldots,{\bf y}_i^{(r_j)}\right).
    \end{align*}
    \item A decoding function $\phi_{i+T}:\underbrace{\mathbb{F}^{n_L}\cup\{*\}\times\ldots\times\mathbb{F}^{n_L}\cup\{*\}}_{i+T+1~{\rm times}}\to\mathbb{F}^{n_{L+1}}$ for each $i\in\mathbb{Z}_{+}$ is used by node $r_{L+1}$ at time $i+T$ to estimate ${\bf s_i}$ according to
    \begin{align}
        \hat{\bf s}_i=\phi_{i+T}\left({\bf y}^{(r_{L+1})}_0,{\bf y}^{(r_{L+1})}_1,\ldots,{\bf y}_{i+T}^{(r_{L+1})}\right).
    \end{align}
\end{enumerate}
\label{def:def1}
\end{definition}
\begin{definition}
An erasure sequence is a binary sequence denoted by $e^{\infty}\triangleq\{e_i\}_{i=0}^{\infty}$ where \mbox{$e_i={\bf 1}\{{\rm erasure~occurs~at~time~i}\}$}.
\end{definition}
An $N$-erasure sequence is an erasure sequence $e^{\infty}$ that satisfies $\sum_{l=0}^{\infty}e_l=N$. Alternatively, we denote it as a $N$-deterministic erasure channel. The set of $N$-erasure sequences is denoted by $\Omega_N$. We further denote $N_1,\ldots,N_{L+1}$ deterministic erasure network model as $N_1,\ldots,N_{L+1}$-erasure sequences, each occur on a different channel, where for any $j\in\{1,\ldots,L+1\}$, the maximal number of erasures on channel $(r_{j-1},r_{j)}$ is $N_j$.
\begin{definition}
The mapping $g_{n}:\mathbb{F}^{n}\times \{0,1\}\to \mathbb{F}^{n}\cup\{*\}$ of an erasure channel is defined as
\begin{align}
    g_{n}({\bf x},e)=\begin{cases}
    {\bf x}~~~ {\rm if~}e=0,\\
    * ~~~ {\rm if~}e=1.
    \end{cases}
\end{align}
Denoting with $e^{j+1}\in\Omega_{N_{j+1}}$ an admissible erasure sequence in channel $(r_{j},r_{j+1})$, for any erasure sequence $e^{j,\infty}$ and any $(n_1,\ldots,n_{L+1},k,T)_{\mathbb{F}}$-streaming code, the following input-output relations holds for the erasure channel $(r_j,r_{j+1})$ for each $i\in\mathbb{Z}_+$:
\begin{align*}
    {\bf y}_i^{(r_{j+1})}=g_{n_{j}}\left({\bf x}^{(r_j)}_i,e_i^{j+1}\right)
    \label{eq:g}
\end{align*}

\end{definition}
\begin{definition}
An $(n_1,\ldots,n_{L+1},k,T)_{\mathbb{F}}$-streaming code is said to be $(N_1,N_2,\ldots,N_{L+1})$-achievable if the following holds for any $N_j$-erasure sequence $e^{j,\infty}\in\Omega_{N_j}$ ($j\in\{1,\ldots,L+1\}$), for all $i\in\mathbb{Z}_+$ and all ${\bf s}_i\in\mathbb{F}^k$, we have
\begin{align*}
    \hat{\bf s}_i={\bf s}_i
\end{align*}
where
\begin{align}
\hat{\bf s}_i=\phi_{i+T}\left(g_{n_{L+1}}\left({\bf x}^{(r_L)}_0,e^{L+1}_0\right),\ldots,g_{n_{L+1}}\left({\bf x}^{(r_L)}_{i+T},e^{L+1}_{i+T}\right)\right)
\end{align}
and for previous nodes
\begin{align}
{\bf x}^{(r_{j})}_i=f_i^{(r_{j})}\left(g_{n_{j}}\left({\bf x}^{(r_{j-1})}_0,e^{j}_0\right),\ldots,g_{n_{j}}\left({\bf x}^{(r_{j-1})}_{i},e^{j}_{i}\right)\right).
\end{align}
\label{def:def4}
\end{definition}
\begin{definition}
The rate of an $(n_1,n_2,\ldots,n_{L+1},k,T)_{\mathbb{F}}$-streaming code is $\frac{k}{\max\{n_1,n_2,\ldots,n_{L+1}\}}$.
\label{def:def5}
\end{definition}

If, for a specific $j$, all $N_{l\neq j}=0$ and $N_{j}\neq 0$, then the $L$-node relay network with erasures reduces to a point-to-point packet erasure channel. It was previously shown in \cite{badr2013streaming} that the maximum achievable rate of the point-to-point  packet erasure channel with $N_j=N$ arbitrary erasures and delay of $T$ symbols denoted by $C_{T,N}$ satisfies
\begin{align}
    C_{T,N}=\begin{cases}
    \frac{T-N+1}{T+1}~~~~T\geq N\\
    0~~~~~~~~~~~~{\rm otherwise}.
    \end{cases}
    \label{eq:P2P}
\end{align}

It was further shown that the capacity of the point-to-point channel with $N$ arbitrary erasures and delay of $T$ could be achieved by diagonally interleaving $(T+1,T-N+1)$ MDS code.

In \cite{fong2018optimal}, a three node relay network was analyzed, and the following theorem was shown.
\begin{theorem}[Theorem 1 in \cite{fong2018optimal}]
Fix any $(T,N_1,N_2)$. Recalling that the point-to-point capacity satisfies~(\ref{eq:P2P}), we have
\begin{align}
    C_{T,N_1,N_2}&=\min\left\{C_{T-N_2,N_1},C_{T-N_1,N_2}\right\} \nonumber \\
    &=\begin{cases}
    \frac{T-N_1-N_2+1}{T-\min\{N_1,N_2\}+1} ~~~~T\geq N_1+N_2\\
    0~~~~~~~~~~~~~~~~~~~~~~{\rm otherwise}.
    \end{cases}
    \label{eq:C_3node}
\end{align}

\end{theorem}

\subsection{Main Results}
\label{sec:mainRes}
In this paper we first derive an upper bound for the achievable rate in $L+1$-node relay network.
\begin{theorem}
Assume a path with $L$ relays. For a target overall delay of $T$, where the maximal number of arbitrary erasures in link $(r_{j-1},r_j),~j\in\{1,\ldots,L+1\}$ is $N_j$. The achievable rate is upper bounded by
\begin{align}
    R&\leq \begin{cases}
    \frac{T-\sum_{l=1}^{L+1} N_l+1}{T-\min_{j}\left\{\sum_{l=1, l\neq j}^{L+1} N_l\right\}+1} \triangleq C^{+}_{T,N_1,\ldots,N_{L+1}} ~~~~~ T\geq \sum_{l=1}^{L+1} N_l \\
    0 ~~~~~~~~~~~~~~~~~~~~~~~~~~~~~~~~~~~~~~~~~~~~~~~~~~~ {\rm otherwise}.
    \end{cases}
\end{align}
\label{thm:conv}
\end{theorem}

We then suggest an achievable scheme which achieves the upper bound up to a size of an overhead which is required by the scheme. Denoting
\begin{align}
    n_{\max}\triangleq\max_{j\in\{1,\ldots,L+1\}} \left(T-\sum_{l=1,l\neq {j}}^{L+1} N_l+1\right),
    \label{eq:Tj}
\end{align}
we show the following Theorem.

\begin{theorem}
Assume a link with $L$ relays. For a target overall delay of $T$, where the maximal number of arbitrary erasures in link $(r_j,r_{j+1}),~j\in\{0,\ldots,L\}$ is $N_{j+1}$ and $T\geq \sum_{l=1}^{L+1} N_l$. When $|\mathbb{F}|\geq n_{\max}$, The following rate is achievable.

\begin{align}
    R & \geq \frac{T-\sum_{l=1}^{L+1}N_l+1}{T-\min_j\left\{\sum_{l=1,l\neq j}^{L+1}N_l\right\}+1+\frac{n_{\max}\lceil\log\left(n_{\max}\right)\rceil}{\log(|\mathbb{F}|)}}
\end{align}
where $n_{\max}$ is defined in (\ref{eq:Tj}).
\label{thm:thm3}
\end{theorem}

\begin{remark}
Although the deterministic erasure model is formulated in such a way that link $(r_{j-1},r_j)$ introduces only a finite number of erasures over the discrete timeline, the maximum coding rate remains unchanged for the following sliding window model that can introduce infinite erasures. Every message must be perfectly recovered with delay $T$ as long as the numbers of erasures introduced by in link $(r_{j-1},r_j)$ in every sliding window of size T + 1 do not exceed $N_j$. This is further described in Section~\ref{sec:extToSliding}.
\end{remark}

We recall that for any natural numbers $L$ and $M$, a systematic maximum-distance separable (MDS) $(L+M, L)$-code is characterized by an $L\times M$ parity matrix $\svv{V}^{L\times M}$ where any $L$ columns of $\left[\svv{I}_L~\svv{V}^{L\times M}\right]\in\mathbb{F}^{L\times(L+M)}$
are independent. It is well known that a systematic MDS $(L + M,L)$-code always exists as long as $|\mathbb{F}|\geq L+M$ \cite{macwilliams1977theory}. To simplify notation, we sometimes denote $N_a^b=\sum_{l=a}^b N_l$. We will take all logarithms to base $2$ throughout this paper. We denote the $i$'th element of vector ${\bf x}$ as $x[i]$, or sometime as $[x][i]$.


\subsection{Motivating example}
\label{sec:motExample}
Consider a link with up to $N=2$ arbitrary erasures, where the delay constraint the transmission has to meet is $T=3$ symbols. The capacity of this link according to \eqref{eq:P2P} is $C_{3,2}=2/4$. Now, assume that in fact, this link is a three-node network ($L=1$), where up to $N_1=1$ erasures occur in link $(r_0,r_1)$ and up to $N_2=1$ erasures occur in link $(r_1,r_2)$, where transmission has to be decoded with the same overall delay of $T=3$ symbols. This is depicted in Figure~\ref{fig:threeNodeNetwork}. The capacity of this link according to (\ref{eq:C_3node}) is $C_{3,1,1}=2/3$, which is better than the point-to-point link.

\begin{figure}[ht]
    \begin{center}
        \begin{psfrags}
            \psfrag{s}[][][0.75]{$s=r_0$}
            \psfrag{d}[][][0.75]{$d=r_{2}$}
            \psfrag{r_1}[][][1]{$r_1$}
            \psfrag{r_2}[][][1]{$r_2$}
            \psfrag{r_k}[][][1]{$r_L$}
            \psfrag{N1}[][][1]{$N_1~{\rm Erasures}$}
            \psfrag{N2}[][][1]{$N_2~{\rm Erasures}$}
            \psfrag{Nk}[][][1]{$N_{L+1}~{\rm Erasures}$}
            \psfrag{A}[][][0.75]{${\bf s}_i\in\mathbb{F}^k$}
            \psfrag{B}[][][0.75]{${\bf x}_i^{(r_0)}\in\mathbb{F}^{n_1}$}
            \psfrag{C}[][][0.75]{${\bf y}_i^{(r_1)}\in\mathbb{F}^{n_1}$}
            \psfrag{D}[][][0.75]{${\bf x}_i^{(r_1)}\in\mathbb{F}^{n_2}$}
            \psfrag{E}[][][0.75]{${\bf y}_i^{(r_2)}\in\mathbb{F}^{n_2}$}
            \psfrag{F}[][][0.75]{${\bf x}_i^{(r_2)}\in\mathbb{F}^{n_3}$}
            \psfrag{G}[][][0.75]{${\bf y}_i^{(r_L)}\in\mathbb{F}^{n_{L}}$}
            \psfrag{H}[][][0.75]{${\bf x}_i^{(r_L)}\in\mathbb{F}^{n_{L+1}}$}
            \psfrag{I}[][][0.75]{${\bf y}_i^{(r_{L+1})}\in\mathbb{F}^{n_{L+1}}$}
            \psfrag{J}[][][0.75]{$\hat{\bf s}_{i-T}\in\mathbb{F}^{k}$}
            \includegraphics[width=0.5\columnwidth]{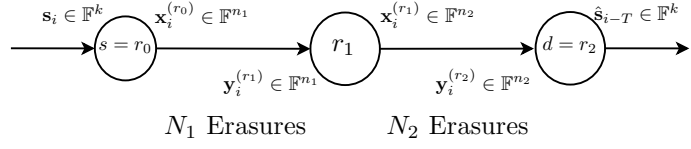}
        \end{psfrags}
    \end{center}
    \caption{A three-node relay network.}
    \label{fig:threeNodeNetwork}
\end{figure}



We start by recalling the coding scheme that was shown in \cite{fong2018optimal} to achieve capacity. The example we show next is the same as provided in \cite{fong2018optimal} for the scenario described above (in which the total delay $T=3$ and the maximal number of erasures in each channel is $N_1=N_2=1$).

Suppose node $s$ transmits two bits $a_i$ and $b_i$ at each discrete time $i\geq 0$ to node $d$ with delay 3. For each time $i$, node $s$ transmits the three-symbol packet ${\bf x}^{(r_0)}_i=\left[a_{i}~~b_{i}~~a_{i-2}+b_{i-1}\right]$ according to Table~\ref{tab:singleRelaySilas} where $a_j=b_j=0$ for any $j<0$ by convention, and the symbols highlighted in the same color are generated by the same block code.

Since channel $(s,r)$ introduces at most $N_1=1$ erasure, each $a_i$ and each $b_i$ can be perfectly recovered by node $r$ by time $i+2$ and time $i+1$ respectively. Therefore, at each time $i$, node $r$ should have recovered $a_{i-2}$ and $b_{i-1}$ perfectly with delays 2 and 1 respectively, and it will re-encode them into another three-symbol packet ${\bf x}^{(r_1)}_{i+1}=\left[b_{i-1}~~a_{i-2}~~b_{i-3}+a_{i-3}\right]$. This is depicted in Table~\ref{tab:singleRelaySilas}. Since $b_{i-3}$, $a_{i-3}$ and $b_{i-3}+a_{i-3}$ are transmitted by node $r$ at time $i-2$, $i-1$ and $i$ respectively, it follows from the fact $N_2=1$ that node $d$ can recover $a_{i-3}$ and $b_{i-3}$ by time $i$. Consequently, this symbol-wise DF strategy achieves a rate of 2/3.


\begin{table}[H]
    \centering
    \begin{subtable}{1\textwidth}
    \centering
    \begin{tabular} {|c|c|c|c|c|c|c|c|}
         \hline
         Time & $i-1$ & $i$& $i+1$ & $i+2$ & $i+3$ & $i+4$ & $i+5$  \\ \hline
         $a_i$ & \lfbox[border-style=solid,border-color=black]{$\textcolor{olive}{a_{i-1}}$} & \lfbox[border-style=dashed,border-color=black]{$\textcolor{red}{a_{i}}$} & \lfbox[border-style=dotted,border-color=black]{$\textcolor{blue}{a_{i+1}}$} & $a_{i+2}$ & $a_{i+3}$ & $a_{i+4}$ & $a_{i+5}$\\ \hline
         $b_i$ & $b_{i-1}$ & \lfbox[border-style=solid,border-color=black]{$\textcolor{olive}{b_{i}}$} & \lfbox[border-style=dashed,border-color=black]{$\textcolor{red}{b_{i+1}}$} & \lfbox[border-style=dotted,border-color=black]{$\textcolor{blue}{b_{i+2}}$} & $b_{i+3}$ & $b_{i+4}$ & $b_{i+5}$ \\ \hline
         $a_{i-2}+b_{i-1}$ & $a_{i-3}+b_{i-2}$ & $a_{i-2}+b_{i-1}$ & \lfbox[border-style=solid,border-color=black]{$\textcolor{olive}{a_{i-1}+b_{i}}$} & \lfbox[border-style=dashed,border-color=black]{$\textcolor{red}{a_{i}+b_{i+1}}$} & \lfbox[border-style=dotted,border-color=black]{$\textcolor{blue}{a_{i+1}+b_{i+2}}$} & $a_{i+2}+b_{i+3}$ & $a_{i+3}+b_{i+4}$\\
         \hline
    \end{tabular}
    \caption{Symbols transmitted by the source node $s$.}
   \end{subtable}
   \begin{subtable}{1\textwidth}
   \centering
    \begin{tabular} {|c|c|c|c|c|c|c|c|}
         \hline
         Time & $i-1$ & $i$& $i+1$ & $i+2$ & $i+3$ & $i+4$ & $i+5$ \\ \hline
         $b_{i-1}$ & $b_{i-2}$ & $b_{i-1}$ & \lfbox[border-style=solid,border-color=black]{$\textcolor{olive}{b_{i}}$} &
         \lfbox[border-style=dashed,border-color=black]{$\textcolor{red}{b_{i+1}}$} & \lfbox[border-style=dotted,border-color=black]{$\textcolor{blue}{b_{i+2}}$} & $b_{i+3}$ & $b_{i+4}$ \\ \hline
         $a_{i-2}$ & $a_{i-3}$ & $a_{i-2}$ & $\textcolor{olive}{a_{i-1}}$ & \lfbox[border-style=solid,border-color=black]{$\textcolor{red}{a_{i}}$} & \lfbox[border-style=dashed,border-color=black]{$\textcolor{blue}{a_{i+1}}$} & \lfbox[border-style=dotted,border-color=black]{$a_{i+2}$} & $a_{i+3}$ \\ \hline
          $a_{i-3}+b_{i-3}$ & $a_{i-4}+b_{i-4}$ & $a_{i-3}+b_{i-3}$ & $a_{i-2}+b_{i-2}$ & $a_{i-1}+b_{i-1}$ & \lfbox[border-style=solid,border-color=black]{$a_{i}+b_{i}$} & \lfbox[border-style=dashed,border-color=black]{$a_{i+1}+b_{i+1}$} & \lfbox[border-style=dotted,border-color=black]{$a_{i+2}+b_{i+3}$} \\ \hline
    \end{tabular}
    \caption{Symbols transmitted by relay node $r$.}
     \end{subtable}
       \begin{subtable}{1\textwidth}
   \centering
    \begin{tabular} {|c|c|c|c|c|c|c|c|}
         \hline
         Time & $i-1$ & $i$& $i+1$ & $i+2$ & $i+3$ & $i+4$ & $i+5$ \\ \hline
         $a_{i-3}$ & $a_{i-4}$ & $a_{i-3}$ & $a_{i-2}$ & $\textcolor{olive}{a_{i-1}}$ & \lfbox[border-style=solid,border-color=black]{$\textcolor{red}{a_{i}}$} & \lfbox[border-style=dashed,border-color=black]{$\textcolor{blue}{a_{i+1}}$} & \lfbox[border-style=dotted,border-color=black]{$a_{i+2}$}\\ \hline
         $b_{i-3}$ & $b_{i-4}$ & $b_{i-3}$ & $b_{i-2}$ & $b_{i-1}$ & \lfbox[border-style=solid,border-color=black]{$\textcolor{olive}{b_{i}}$} & \lfbox[border-style=dashed,border-color=black]{$\textcolor{red}{b_{i+1}}$} & \lfbox[border-style=dotted,border-color=black]{$\textcolor{blue}{b_{i+2}}$}\\ \hline
    \end{tabular}
    \caption{Estimates constructed by the destination node $d$.}
     \end{subtable}
     \caption{A state-less symbol-wise DF strategy for a single relay. Symbols belong to the same code used by the source ($s$) and symbols belong to the same code used by the relay ($r$) are marked with a frame with the same style.}
    \label{tab:singleRelaySilas}
\end{table}

An important feature of this code is the fact that it is a state-less code, i.e., the structure of the code does not depend on the specific erasure pattern at any of the segments. However, if another relay is to be considered (i.e., the destination is now replaced with relay $r_2$), assuming up to $N_3=1$ erasures in link $(r_2,d)$, the minimal delay required to support rate $2/3$ is $T=5$. The network of interest is depicted in Figure~\ref{fig:fourNodeNetwork}.

\begin{figure}[ht]
    \begin{center}
        \begin{psfrags}
            \psfrag{s}[][][0.75]{$s=r_0$}
            \psfrag{d}[][][0.75]{$d=r_{3}$}
            \psfrag{r_1}[][][1]{$r_1$}
            \psfrag{r_2}[][][1]{$r_2$}
            \psfrag{r_k}[][][1]{$r_L$}
            \psfrag{N1}[][][1]{$N_1~{\rm Erasures}$}
            \psfrag{N2}[][][1]{$N_2~{\rm Erasures}$}
            \psfrag{Nk}[][][1]{$N_{3}~{\rm Erasures}$}
            \psfrag{A}[][][0.75]{${\bf s}_i\in\mathbb{F}^k$}
            \psfrag{B}[][][0.75]{${\bf x}_i^{(r_0)}\in\mathbb{F}^{n_1}$}
            \psfrag{C}[][][0.75]{${\bf y}_i^{(r_1)}\in\mathbb{F}^{n_1}$}
            \psfrag{D}[][][0.75]{${\bf x}_i^{(r_1)}\in\mathbb{F}^{n_2}$}
            \psfrag{E}[][][0.75]{${\bf y}_i^{(r_2)}\in\mathbb{F}^{n_2}$}
            \psfrag{F}[][][0.75]{${\bf x}_i^{(r_2)}\in\mathbb{F}^{n_3}$}
            \psfrag{G}[][][0.75]{${\bf y}_i^{(r_L)}\in\mathbb{F}^{n_{L}}$}
            \psfrag{H}[][][0.75]{${\bf x}_i^{(r_L)}\in\mathbb{F}^{n_{L+1}}$}
            \psfrag{I}[][][0.75]{${\bf y}_i^{(r_{3})}\in\mathbb{F}^{n_{L+1}}$}
            \psfrag{J}[][][0.75]{$\hat{\bf s}_{i-T}\in\mathbb{F}^{k}$}
            \includegraphics[width=0.7\columnwidth]{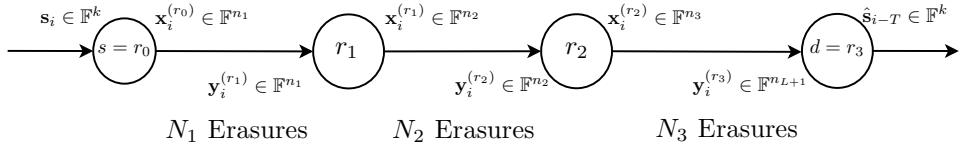}
        \end{psfrags}
    \end{center}
    \caption{A four-node relay network.}
     \label{fig:fourNodeNetwork}
\end{figure}

\begin{table}[H]
    \centering
    \begin{tabular} {|c|c|c|c|c|c|c|c|}
         \hline
         Time & $i-1$ & $i$& $i+1$ & $i+2$ & $i+3$ & $i+4$ & $i+5$  \\ \hline
         $a_{i-3}$ & $a_{i-4}$ & $a_{i-3}$ & $a_{i-2}$ & \lfbox[border-style=solid,border-color=black]{$\textcolor{olive}{a_{i-1}}$} & \lfbox[border-style=dashed,border-color=black]{$\textcolor{red}{a_{i}}$} & $a_{i+1}$ & $a_{i+2}$\\ \hline
         $b_{i-3}$ & $b_{i-4}$ & $b_{i-3}$ & $b_{i-2}$ & $b_{i-1}$ & \lfbox[border-style=solid,border-color=black]{$\textcolor{olive}{b_{i}}$} & \lfbox[border-style=dashed,border-color=black]{$\textcolor{red}{b_{i+1}}$} & $b_{i+2}$ \\ \hline
         $a_{i-5}+b_{i-4}$ & $a_{i-6}+b_{i-5}$ & $a_{i-5}+b_{i-4}$ & $a_{i-4}+b_{i-3}$ & $a_{i-3}+b_{i-2}$ & $a_{i-2}+b_{i-1}$ & \lfbox[border-style=solid,border-color=black]{$\textcolor{olive}{a_{i-1}+b_{i}}$} & \lfbox[border-style=dashed,border-color=black]{$\textcolor{red}{a_{i}+b_{i+1}}$}\\
         \hline
    \end{tabular}
    \caption{Symbols transmitted by node $r_2$ when trying to extend the state-less symbol-wise DF strategy. Symbols belong to the same code are denoted with frame with the same style.}
    \label{tab:singleRelaySilasAddRelay}
\end{table}

This can be seen since basically using the coding scheme in Table~\ref{tab:singleRelaySilas}, relay $r_2$ can be viewed as source $s$ with delay of $T=3$ packets (essentially, the delay of different symbols is ``reset''). Due to causality, relay $r_2$ can only use the coding scheme of the sender, depicted in Table~\ref{tab:singleRelaySilasAddRelay}. It can be viewed that when $N_3=1$, symbol $a_i$ is guaranteed to be recovered only at time $i+5$.

In this paper, we suggest a state-dependent scheme, i.e., a scheme in which the order of symbols in each block code used (and thus the order in each diagonal) is set according to the erasure pattern in the previous link. The order of the symbols is transmitted to the receiver to allow decoding. Hence, additional overhead is required. We first show an example of the suggested scheme to the three-node network and then show how to extend it to a four-node network.

In the proposed scheme, the source $r_0$ uses the same code suggested in \cite{fong2018optimal}, i.e., a $(3,2)$ MDS code that can recover any arbitrary single erasure in a delay of two symbols combined with diagonal interleaving (i.e., the code is applied on the diagonals). For each time $i$, node $s$ transmits the three-symbol packet (symbols belong to the same block code are marked with the same color).

\begin{table}[ht]
\begin{center}
\begin{tabular} {|c|c|c|c|c|c|c|}
     \hline
     Time & $i-1$ & $i$& $i+1$ & $i+2$ & $i+3$ & $i+4$ \\ \hline
     Header & 123 & \tikzmark{topD}{123} & 123 & 123 & 123 & 123 \\ \hline
     $a_i$ & $\textcolor{olive}{a_{i-1}}$ & $\textcolor{red}{a_{i}}$ & $\textcolor{blue}{a_{i+1}}$ & $a_{i+2}$ & $a_{i+3}$ & $a_{i+4}$ \\ \hline
     $b_i$ & $b_{i-1}$ & $\textcolor{olive}{b_{i}}$ & $\textcolor{red}{b_{i+1}}$ & $\textcolor{blue}{b_{i+2}}$ & $b_{i+3}$ & $b_{i+4}$ \\ \hline
     $a_{i-2}+$ & $a_{i-3}+$ & $a_{i-2}+$ & $\textcolor{olive}{a_{i-1}+}$ & $\textcolor{red}{a_{i}+}$ & $\textcolor{blue}{a_{i+1}+}$ & $a_{i+2}+$ \\
     $b_{i-1}$ & $b_{i-2}$ & \tikzmark{rightD}{b_{i-1}} & $\textcolor{olive}{b_{i}}$ & $\textcolor{red}{b_{i+1}}$ & $\textcolor{blue}{b_{i+2}}$ & $b_{i+3}$ \\ \hline
\end{tabular}
\DrawVLine[black, thick, opacity=0.2]{topD}{rightD}
\end{center}
\caption{Transmission of the Source in channel $(r_0,r_1)$. Symbols belong to the same code are marked with the same color.}
\label{table:tabR0}
\end{table}

When there are no erasures, relay $r_1$ uses the same code as the source $r_0$ \emph{while delaying it by one symbol}, i.e. relay $r_1$ sends the following three-symbol packet ${\bf x}^{(r_1)}_{i+1}=\left[a_i,b_{i},a_{i-2}+b_{i-1}\right]$.

If an erasure occurred, the relay would send any available symbols (per diagonal) in the received order until it can decode the information symbols from this block code. Then, the erased symbols will be sent. For example, assuming that a symbol at time $i$ is erased in link $(r_0,r_1)$, relay $r_1$ will send ${\bf x}^{(r_1)}_{i+1}=\left[b_{i+1},b_i,a_{i-2}+b_{i-1}\right]$ and ${\bf x}^{(r_1)}_{i+2}=\left[a_{i+1},a_i,a_{i-1}+b_{i}\right]$ as depicted in Table \ref{table:tabR1} below. Again, symbols belong to the same block code are marked with the same color. Further, headers which are different than the ones used by $r_0$ are marked with a frame.




\begin{table}[ht]
\begin{center}
\begin{tabular} {|c|c|c|c|c|c|c|}
     \hline
     Time & $i-1$ & $i$& $i+1$ & $i+2$ & $i+3$ & $i+4$ \\ \hline
      Header & 123 & 123 & \fbox{223} & \tikzmark{topD}{\fbox{113}} & 123 & 123 \\ \hline
      & $a_{i-2}$ & $\textcolor{olive}{a_{i-1}}$ & $\textcolor{red}{\bf{b_{i+1}}}$ & $\textcolor{blue}{a_{i+1}}$ & $a_{i+2}$ & $a_{i+3}$ \\ \hline
      & $b_{i-2}$ & $b_{i-1}$ & $\textcolor{olive}{b_{i}}$ & $\textcolor{red}{\bf{a_{i}}}$ & $\textcolor{blue}{b_{i+2}}$ & $b_{i+3}$ \\ \hline
     & $a_{i-4}+$ & $a_{i-3}+$ & $a_{i-2}+$ & $\textcolor{olive}{a_{i-1}+}$ & $\textcolor{red}{a_{i}+}$ & $\textcolor{blue}{a_{i+1}+}$ \\
     & $b_{i-3}$ & $b_{i-2}$ & $b_{i-1}$ & \tikzmark{rightD}{\textcolor{olive}{b_{i}}} & $\textcolor{red}{b_{i+1}}$ & $\textcolor{blue}{b_{i+2}}$ \\ \hline
\end{tabular}
\DrawVLine[black, dashed, opacity=0.8]{topD}{rightD}
\end{center}
\caption{Transmission of relay $r_1$ in channel $(r_1,r_2)$, given that symbol $i$ was erased when transmitted in link $(r_0,r_1)$. Symbols belong to the same code are marked with the same color.}
\label{table:tabR1}
\end{table}

We note that the erasure in time $i$ in link $(r_0,r_1)$ caused a change in the order of the symbols in packets $i+1$ and $i+2$. Denoting the order of symbols sent from $r_0$ in each code (alternatively on each diagonal) as $[1,2,3]$, in this example, the header of each packet is composed of the location of the symbols from each block code (the order of the symbols in each diagonal) as they would appear in the code used by $r_0$.

We show next that the code used by each relay on each diagonal can be viewed as a punctured code of a code that is used by the ``bottleneck'' relay. i.e., the relay with the smallest rate. Hence, we note that it can be viewed as if the suggested coding scheme only changes the order of symbols per diagonal (taken from the code with the lowest rate), and it does not add or remove symbols. Therefore, using a single index per symbol suffices to allow recovery at each destination.

In our example, as can bee seen in Table~\ref{table:tabR1}, at time $i$ relay $r_1$ sends ${\bf x}_{i}^{(r_1)}=[a_{i-1}~~b_{i-1}~~a_{i-3}+b_{i-2}]$ with header $``123"$ indicating that the first symbol is the first symbol (marked with underline) in the code $\left[\underline{a_{i-1}}~~b_{i}~~a_{i-1}+b_{i}\right]$, the second symbol is the second symbol in the code $\left[a_{i-2}~~\underline{b_{i-1}}~~a_{i-2}+b_{i-1}\right]$ and the third symbol is the code $\left[a_{i-3}~~b_{i-2}~~\underline{a_{i-3}+b_{i-2}}\right]$.

However, following the erasure occurred at link $(r_0,r_1)$ at time $i$, relay $r_1$ can not send symbol $a_i$ at time $i+1$ as planned (as if there was no erasure). However, it can send $b_{i+1}$ which was not erased. With respect to the second symbol, we note that since $a_{i-1}$ and $a_{i-1}+b_i$ were received, $b_i$ can be recovered and used. Similarly, since $a_{i-2}$ and $b_{i-1}$ were received, symbol $a_{i-2}+b_{i-1}$ can be generated and used. Thus, at time $i+1$ the relay can send ${\bf x}_{i+1}^{(r_1)}=[b_{i+1}~~b_{i}~~a_{i-2}+b_{i-1}]$. To indicate the change in order of the symbols used in the first code, the header is changed to $``223"$ which should be indicating that the first symbol is the second symbol in the code associated with this diagonal.

At time $i+2$, the relay can recover and send $a_i$. Hence, it sends ${\bf x}_{i+2}^{(r_1)}=[a_{i+1}~~a_{i}~~a_{i-1}+b_{i}]$ with header $``113"$ indicating now that the second symbol is the first symbol in the code associated with this diagonal. It can be easily verified that the destination can recover the original data at a delay of $T=3$ (assuming any single arbitrary erasure in the link between the relay and destination).





This concept can be applied to additional relays if they exist. For example consider four-node network ($L=2$). The transmission scheme on the next relay $r_2$ (in this specific example), is the same as the transmission scheme of the first relay, i.e., in case there is no erasure, transmit (on each diagonal) the symbols in the same order as received, delayed by one symbol (i.e., a total delay of two symbols from the sender). If an erasure occurred before $r_2$ decoded the information symbols, transmit the available symbols (again, it is guaranteed that there will be enough symbols). When the information symbols can be decoded, transmit the erased symbols.

For example, in case the symbol transmitted from relay $r_1$ to $r_2$ at time $i+2$ is erased, the suggested transmission scheme of relay $r_2$ is given in Table~\ref{table:tabR2} below. Again, symbols belong to the same block code are marked with the same color. Further, headers that are different than the ones used by the relay $r_1$ are marked with a frame.

\begin{table}[ht]
\begin{center}
\begin{tabular} {|c|c|c|c|c|c|c|}
     \hline
     Time & $i-1$ & $i$& $i+1$ & $i+2$ & $i+3$ & $i+4$ \\ \hline
     Header & 123 & 123 & 123 & 223 & \fbox{213} & \fbox{113} \\ \hline
     & $a_{i-3}$ & $a_{i-2}$ & $\textcolor{olive}{a_{i-1}}$ & $\textcolor{red}{b_{i+1}}$ & $\textcolor{blue}{b_{i+2}}$ & $a_{i+2}$ \\ \hline
     & $b_{i-3}$ & $b_{i-2}$ & $b_{i-1}$ & $\textcolor{olive}{b_i}$ & $\textcolor{red}{a_{i}}$ & $\textcolor{blue}{a_{i+1}}$
     \\ \hline
     & $a_{i-5}+$ & $a_{i-4}+$ & $a_{i-3}+$ & $a_{i-2}+$ & $\textcolor{olive}{a_{i-1}}$ & $\textcolor{red}{a_{i}+}$ \\
     & $b_{i-4}$ & $b_{i-3}$ & $b_{i-2}$ & $b_{i-1}$ & $\textcolor{olive}{b_i}$ & $\textcolor{red}{b_{i+1}}$ \\ \hline
\end{tabular}
\end{center}
\caption{Transmission of relay $r_2$ in channel $(r_2,r_3)$, given that symbol $i+2$ was erased when transmitted in link $(r_1,r_2)$. Symbols belong to the same code are marked with the same color.}
\label{table:tabR2}
\end{table}

We note that the erasure in time $i+2$ in link $(r_1,r_2)$ caused a change in order of the symbols in packets $i+3$ and $i+4$ (the order of symbols in time $i+2$ is not the original order, yet it is the same as was transmitted from $r_1$ at time $i+1$).

Since the same $(3,2)$ code is used (with a different order of symbols which is communicated to the receiver), it can be easily verified that each packet can be recovered up to delay of $T=4$ symbols for any arbitrary erasure in the link between the relay and the destination.


In this example, the maximal size of the header is three numbers, each taken from $\{1,2,3\}$. Hence, its maximal size of the header is $3\log(3)$ bits. Since the block code used in each link transmits two bits using three bits in every channel use, we conclude that the scheme achieves a rate of
\begin{align}
    R&=\frac{2}{3+3\lceil\log(3)\rceil}. 
\end{align}

We show next that this idea can be extended to transmitting symbols taken from any field $\mathbb{F}$ thus in general, the achievable rate when $T=4$ and $N_1=N_2=N_3=1$ is

\begin{align}
    R&=\frac{2\cdot\log(|\mathbb{F}|)}{3\cdot\log(|\mathbb{F}|)+3\lceil\log(3)\rceil} \nonumber \\
    &=\frac{2}{3+\frac{3\lceil\log(3)\rceil}{\log(|\mathbb{F}|)}}, 
\end{align}
which approaches $2/3$ as the field size increases. As we further show next, the upper bound for this scenario is indeed $2/3$.

\section{Proof of the upper bound}
\label{sec:converse}
Fix any $(N_1,\ldots,N_{L+1},T)$. Suppose we are given an $(N_1,\ldots,N_{L+1})$-achievable $(n_1,\ldots,n_{j+1},k,T)_{\mathbb{F}}$-streaming code for some $n_1,\ldots,n_{j+1},k$ and $\mathbb{F}$. Our goal is to show that
\begin{align}
    \frac{k}{\max\left\{n_1,\ldots,n_{L+1}\right\}}&\leq\min_{j\in\{1,2,\ldots,L+1\}}\frac{T-\sum_{l=1}^{L+1} N_l+1}{T-\sum_{l=1,l\neq j}^{L+1} N_l +1}.
\end{align}

To this end, we let $\{{\bf s}_i\}_{i\in\mathbb{Z}_+}$ be i.i.d. random variables where ${\bf s}_0$ is uniform on $\mathbb{F}^k$. Since the $(n_1,\ldots,n_{j+1},k,T)_{\mathbb{F}}$-streaming code is $(N_1,\ldots,N_{L+1})$-achievable, it follows from Definition~\ref{def:def4} that
\begin{align}
    H\left({\bf s}_i\given[\Big]{\bf y}^{(L+1)}_0,{\bf y}^{(L+1)}_1,\ldots,{\bf y}^{(L+1)}_{i+T}\right)=0
    \label{eq:entropy}
\end{align}
for any $i\in\mathbb{Z}_{+}$ and any $e^{j,\infty}\in\Omega_{N_j}$. Consider the two cases.

{\bf Case $T<\sum_{l=1}^{L+1}N_l$:}

Let $e^{j,\infty}\in\Omega_{N_j}$ be the error sequence on link $(r_{j-1},r_j)$ where $j\in\{1,2,\ldots,L+1\}$ such that
\begin{align}
    e^{j,\infty}_i=\begin{cases}
    1~~~{\rm if}~\sum_{l=1}^{j-1}N_l\leq i \leq \sum_{l=1}^{j}N_l-1\\
    0~~~{\rm otherwise}.
    \end{cases}
    \label{eq:defOfeJ}
\end{align}

We note that (\ref{eq:defOfeJ}) means that ${\bf y}^{(r_1)}_{N_1}$ is the first packet which can be used to recover ${\bf s}_0$ at $r_1$. Further, ${\bf y}^{(r_2})_{N_1+N_2}$ is the first packet which can be used to recover ${\bf s}_0$ at $r_2$. Continuing the transmission across all other relays, it follows that ${\bf y}^{(r_{L+1})}_{\sum_{N_l}}$ is the first packet which can be used to recover ${\bf s}_0$. Since $T<\sum_{l=1}^{L+1}N_l$ it follows the delay constraint can not be met.

Hence, due to (\ref{eq:defOfeJ}) and Definition \ref{def:def1}, we have
\begin{align}
    I\left({\bf s}_0;{\bf y}^{(L+1)}_0,{\bf y}^{(L+1)}_1,\ldots,{\bf y}^{(L+1)}_{T}\right)=0.
    \label{eq:mutualInfo}
\end{align}

Combining (\ref{eq:entropy}), (\ref{eq:mutualInfo}) and the assumption that $T<\sum_{l=1}^{L+1}N_l$, we obtain $H({\bf s}_0)=0.$ Since ${\bf s}_0$ consists of $k$ uniform random variables in $\mathbb{F}$, the only possible value of $k$ is zero, which implies
\begin{align}
    \frac{k}{\max\left\{n_1,\ldots,n_{L+1}\right\}}=0.
\end{align}

{\bf Case $T\geq\sum_{l=1}^{L+1}N_l$:}

The proof follows the footsteps of \cite{fong2018optimal}. We start by generalizing the arguments given in  \cite{fong2018optimal} for the first and second segments to the first and last segments in our case. Then we show how similar techniques can be used to derive a constraint on the code to be used in an intermediate segment.

{\bf First Segment (link $(r_0,r_1)$):}

First we note that for every $i\in\mathbb{Z}_+$, message ${\bf s}_i$ has to be perfectly recovered by node $r_1$ by time $i+T-\sum_{l=2}^{L+1}N_l$ given that ${\bf s}_0,{\bf s}_1,\ldots,{\bf s}_{i-1}$ have been correctly decoded by node $r_1$, or otherwise a length $N_2$ burst erasure from time \mbox{$i+T-\sum_{l=2}^{L+1}N_l+1$} to $i+T-\sum_{l=3}^{L+1}N_l$ introduced on channel $(r_1,r_2)$ followed by a length $N_3$ burst erasure from time $i+T-\sum_{l=3}^{L+1}N_l+1$ to $i+T-\sum_{l=4}^{L+1}N_l$ introduced on channel $(r_2,r_3)$ and so on until a length $N_{L+1}$ burst erasure from time $i+T-N_{L+1}+1$ to $i+T$ would result in a decoding failure for node $r_1$, node $r_2$ and all the nodes up to the destination $r_{L+1}$.

Recalling that $N_2^{L+1}=\sum_2^{L+1}N_l$ it then follows that
\begin{align}
    H\left({\bf s}_i\given[\Big]\left\{{\bf x}_i^{(r_0)},{\bf x}_{i+1}^{(r_0)},\ldots,{\bf x}_{i+T-N_2^{L+1}}^{(r_0)} \right\}\setminus \left\{{\bf x}_{\theta_1},\ldots,{\bf x}_{\theta_{N_1}}\right\},{\bf s}_0,\ldots,{\bf s}_{i-1}\right)=0
    \label{eq:entropyFirstSegment}
\end{align}
for any $i\in\mathbb{Z}_+$ and $N_1$ non-negative integers denoted by $\theta_1,\ldots,\theta_{N_1}$. We note that the following holds (by assuming, for example, the last $N_1$ symbols in every window of $T-N_2^{L+1}$ symbols starting time $i=0$ are erased):
\begin{align}
    &H\left({\bf s}_0\given[\Big]\left\{{\bf x}_0^{(r_0)},{\bf x}_{1}^{(r_0)},\ldots,{\bf x}_{T-N_2^{L+1}-N_1}^{(r_0)} \right\}\right)=0 \nonumber \\
    &H\left({\bf s}_1\given[\Big]\left\{{\bf x}_1^{(r_0)},{\bf x}_{2}^{(r_0)},\ldots,{\bf x}_{T-N_2^{L+1}-N_1}^{(r_0)},{\bf x}_{1\cdot(T-N_2^{L+1}+1)}^{(r_0)} \right\},{\bf s}_0\right)=0 \nonumber \\
    &~~~~~~~~~~~~~~~~~~~~~~~~~~~~~~\vdots \nonumber \\
    &H\left({\bf s}_{T-N_2^{L+1}-N_1}\given[\Big]\left\{{\bf x}_{T-N_2^{L+1}-N_1}^{(r_0)},{\bf x}_{1\cdot(T-N_2^{L+1}+1)}^{(r_0)},{\bf x}_{1\cdot(T-N_2^{L+1}+1)+1}^{(r_0)},\ldots, {\bf x}_{1\cdot(T-N_2^{L+1}+1)+T-N_2^{L+1}-N_1-1}^{(r_0)}
    \right\},{\bf s}_0,\ldots,{\bf s}_{T-N_2^{L+1}-N_1-1}\right) \nonumber \\
    &=0 \nonumber \\
    &H\left({\bf s}_{T-N_2^{L+1}-N_1+1}\given[\Big]\left\{{\bf x}_{1\cdot(T-N_2^{L+1}+1)}^{(r_0)},\ldots, {\bf x}_{1\cdot(T-N_2^{L+1}+1)+T-N_2^{L+1}-N_1}^{(r_0)}
    \right\},{\bf s}_0,\ldots,{\bf s}_{T-N_2^{L+1}-N_1}\right)=0 \nonumber \\
    &H\left({\bf s}_{T-N_2^{L+1}-N_1+2}\given[\Big]\left\{{\bf x}_{1\cdot(T-N_2^{L+1}+1)}^{(r_0)},\ldots, {\bf x}_{1\cdot(T-N_2^{L+1}+1)+T-N_2^{L+1}-N_1}^{(r_0)}
    \right\},{\bf s}_0,\ldots,{\bf s}_{T-N_2^{L+1}-N_1+1}\right)=0 \nonumber \\
    &~~~~~~~~~~~~~~~~~~~~~~~~~~~~~~\vdots \nonumber \\
    &H\left({\bf s}_{T-N_2^{L+1}-1}\given[\Big]\left\{{\bf x}_{1\cdot(T-N_2^{L+1}+1)}^{(r_0)},\ldots, {\bf x}_{1\cdot(T-N_2^{L+1}+1)+T-N_2^{L+1}-N_1}^{(r_0)}
    \right\},{\bf s}_0,\ldots,{\bf s}_{T-N_2^{L+1}-2}\right)=0 \nonumber \\
    &H\left({\bf s}_{T-N_2^{L+1}}\given[\Big]\left\{{\bf x}_{1\cdot(T-N_2^{L+1}+1)}^{(r_0)},\ldots, {\bf x}_{1\cdot(T-N_2^{L+1}+1)+T-N_2^{L+1}-N_1}^{(r_0)}
    \right\},{\bf s}_0,\ldots,{\bf s}_{T-N_2^{L+1}-1}\right)=0 \nonumber \\
    &~~~~~~~~~~~~~~~~~~~~~~~~~~~~~~\vdots
\end{align}

Using the chain rule we have the following for each $j\in\mathbb{N}$
\begin{align}
    H\left({\bf s}_0,{\bf s}_1,\ldots,{\bf s}_{T-N_2^{L+1}+(j-1)(T-N_2^{L+1}+1)}\given[\Big]  \left\{{\bf x}^{(r_0)}_{m\cdot(T-N_2^{L+1}+1)},{\bf x}^{(r_0)}_{1+m\cdot(T-N_2^{L+1}+1)},
    \ldots, {\bf x}^{(r_0)}_{T-N_1-N_2^{L+1}+m\cdot(T-N_2^{L+1}+1)}\right\}_{m=0}^j \right)
    =0.
    \label{eq:outcomeFirst}
\end{align}

Alternatively we note that for all $q\in\mathbb{Z}_+$,
\begin{align}
    &\left| \left\{q,q+1,\dots,q+T-N_2^{L+1}\right\}\bigcap\right. \nonumber \\
    & \left. \left\{m\cdot(T-N_2^{L+1}+1),1+m\cdot(T-N_2^{L+1}+1)\ldots,T-N_1-N_2^{L+1}+m\cdot(T-N_2^{L+1}+1)\right\}_{m=0}^j \right| \nonumber \\
    &=T-N_1-N_2^{L+1}+1.
    \label{eq:countFirst}
\end{align}
Hence, (\ref{eq:outcomeFirst}) also follows from (\ref{eq:entropyFirstSegment}), (\ref{eq:countFirst}) and the chain rule.

Therefore, following the arguments in \cite{fong2018optimal} which we recall here and in Appendix~\ref{sec:app1}, we note that the \eqref{eq:outcomeFirst} involves $j(T-N_2^{L+1}+1)$ source messages and $(j+1)(T-N_1^{L+1}+1)$ source packets. Therefore, the $(N_1,\ldots,N_{L+1})$-achievable $(n_1,\ldots,n_{L+1}k,T)_{\mathbb{F}}$-streaming code restricted to channel $(r_0,r_1)$ can be viewed as a point-to-point streaming code with rate $k/n_1$ and delay $T-N_2^{L+1}$ which can correct any $N_1$ erasures. In particular the point-to-point code can correct the periodic erasure sequence $\tilde{e}^{\infty}$ depicted in Figure~\ref{fig:Periodic1}, which is formally defined as
\begin{align}
    \tilde{e}_i=\begin{cases}
    0~~~{\rm if}~0\leq i\mod(T-N_2^{L+1}+1)\leq T-N_1^{L+1} \\
    1~~~{\rm otherwise}
    \end{cases}
\end{align}
\begin{figure}[h]
    \begin{center}
        \begin{psfrags}
            \psfrag{d1}[][][1]{$\ldots$}
            \psfrag{a}[][][1]{$T-N_1^{L+1}+1$}
            \psfrag{b}[][][1]{$N_1$}
            \includegraphics[width=0.5\columnwidth]{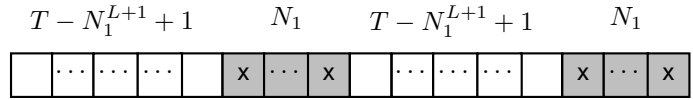}
        \end{psfrags}
    \end{center}
    \caption{A periodic erasure sequence with period $T-N_2^{L+1}+1$.}
    \label{fig:Periodic1}
\end{figure}

By standard arguments which are rigorously elaborated in Appendix~\ref{sec:app1}, we conclude that
\begin{align}
    \frac{k}{n_1}&\leq \frac{T-\sum_{l=1}^{L+1} N_l+1}{T-\sum_{l=2}^{L+1} N_l +1}\nonumber \\
    &=C_{T-\sum_{l=2}^{L+1} N_l,N_1}.
    \label{eq:rateOfFirst}
\end{align}

{\bf Last Segment (link $(r_{L},r_{L+1})$):}

In addition, for every $i\in\mathbb{Z}_+$, message ${\bf s}_i$ has to be perfectly recovered from
\begin{align}
    \left\{{\bf y}^{(r_{L+1})}_{i+N_{1}^{L}},{\bf y}^{(r_{L+1})}_{i+N_{1}^{L}+1},\ldots,{\bf y}^{(r_{L+1})}_{i+T}\right\}
\end{align}
by node $r_{L+1}$ given that ${\bf s}_0,{\bf s}_1,\ldots,{\bf s}_{i-1}$ have been correctly decoded by node $r_{L+1}$, or otherwise a length $N_1$ burst erasure from time $i$ to $i+N_1-1$ induced on channel $(r_0,r_1)$ followed by a length $N_2$ burst erasure from time $i+N_1$ to $i+N_1+N_2-1$ induced on channel $(r_1,r_2)$ and so on until a length $N_L$ burst erasure from time $i+\sum_{l=1}^{L-1}N_l$ to $i+\sum_{l=1}^{L}N_l-1$ induced on channel $(r_{L-1},r_L)$ would result in a decoding failure for node $r_{L+1}$.

It follows that
\begin{align}
    H\left({\bf s}_i\given[\Big]\left\{{\bf x}_{i+N_1^{L}}^{(r_{L})},{\bf x}_{i+N_1^{L}+1}^{(r_{L})},\ldots,{\bf x}_{i+T}^{(r_L)} \right\}\setminus  \left\{{\bf x}_{\theta_1},\ldots,{\bf x}_{\theta_{N_{L+1}}}\right\},{\bf s}_0,\ldots,{\bf s}_{i-1}\right)=0
    \label{eq:LastBasic}
\end{align}

for any $i\in\mathbb{Z}_+$ and $N_{L+1}$ non-negative integers denoted by $\theta_1,\ldots,\theta_{N_{L+1}}$. We note that for all $q\in\mathbb{Z}_+$,
\begin{align}
    &\left| \left\{q+N_1^{L},q+N_1^{L}+1,\dots,q+T\right\}\bigcap\right. \nonumber \\
    & \left. \left\{N_1^{L}+N_{L+1}+m\cdot(T-N_1^{L}+1),N_1^{L}+N_{L+1}+1+m\cdot(T-N_1^{L}+1)\ldots,T+m\cdot(T-N_1^{L}+1)\right\}_{m=0}^j \right| \nonumber \\
    &=T-N_1^{L}-N_{L+1}+1
    \label{eq:countLast}
\end{align}
Using (\ref{eq:LastBasic}), (\ref{eq:countLast}) and the chain rule, we have
\begin{align}
    &H\left({\bf s}_0,\ldots,{\bf s}_{T-N_1^{L}+(j-1)(T-N_1^{L}+1)}\given[\Big]  \left\{{\bf x}^{(r_{L})}_{N_1^{L}+N_{L+1}+m\cdot(T-N_1^{L}+1)},{\bf x}^{(r_{L})}_{N_1^{L}+N_{L+1}+1+m\cdot(T-N_1^{L}+1)},
    \ldots, {\bf x}^{(r_{L})}_{T+m\cdot(T-N_1^{L}+1)}\right\}_{m=0}^j \right) \nonumber \\
    &=0.
    \label{eq:outcomeLast}
\end{align}

Therefore, following the arguments in \cite{fong2018optimal} which we recall here and in Appendix~\ref{sec:app1}, we note that the \eqref{eq:outcomeLast} involves $j(T-N_1^{L}+1)$ source messages and $(j+1)(T-N_1^{L+1}+1)$ source packets. Therefore, the $(N_1,\ldots,N_{L+1})$-achievable $(n_1,\ldots,n_{L+1}k,T)_{\mathbb{F}}$-streaming code restricted to channel $(r_L,r_{L+1})$ can be viewed as a point-to-point streaming code with rate $k/n_{L+1}$ and delay $T-N_1^{L}$ which can correct any $N_{L+1}$ erasures. In particular the point-to-point can correct the periodic erasure sequence $\hat{e}^{\infty}$ depicted in Figure~\ref{fig:Periodic2}, which is formally defined as
\begin{align}
    \tilde{e}_i=\begin{cases}
    0~~~{\rm if}~0\leq i\mod(T-N_1^{L}+1)\leq T-N_1^{L+1} \\
    1~~~{\rm otherwise}
    \end{cases}
\end{align}
\begin{figure}[h]
    \begin{center}
        \begin{psfrags}
            \psfrag{d1}[][][1]{$\ldots$}
            \psfrag{a}[][][1]{$T-N_1^{L+1}+1$}
            \psfrag{b}[][][1]{$N_{L+1}$}
            \includegraphics[width=0.5\columnwidth]{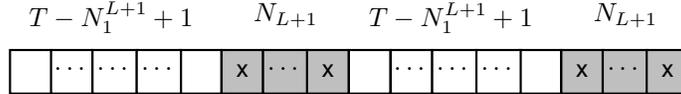}
        \end{psfrags}
    \end{center}
    \caption{A periodic erasure sequence with period $T-N_1^{L}+1$.}
    \label{fig:Periodic2}
\end{figure}

By standard arguments which are rigorously elaborated in Appendix~\ref{sec:app1}, we conclude that
\begin{align}
    \frac{k}{n_{L+1}}&\leq \frac{T-\sum_{l=1}^{L+1} N_l+1}{T-\sum_{l=1}^L N_l +1} \nonumber \\
    &=C_{T-\sum_{l=1}^L N_l,N_{L+1}}.
    \label{eq:rateOfLast}
\end{align}

{\bf The $j$'th Segment (link $(r_{j-1},r_j)$):}

Now, when considering a $j$'th segment (the channel $(r_{j-1},r_j)$), we show again that any $(N_1,\ldots,N_{L+1})$ achievable $(n_1,\ldots n_{L+1},k,T)_{\mathbb{F}}$ code restricted to the channel $(r_{j-1},r_{j})$ can be viewed as a point-to-point code which should handle any $N_{J}$ erasures with a delay which we define next.

Combining the arguments given above we note that first, for every $i\in\mathbb{Z}_+$, message ${\bf s}_i$ has to be perfectly recovered by node $r_{j}$ by time $i+T-\sum_{l=j+1}^{L+1} N_l$ given that ${\bf s}_0,{\bf s}_1,\ldots,{\bf s}_{i-1}$ have been correctly decoded by node $r_{j}$, or otherwise a length $N_{j+1}$ burst from time $i+T-\sum_{l=j+1}^{L+1} N_l+1$ to $i+T-\sum_{l=j+2}^{L+1} N_l$ introduced on channel $(r_{j},r_{j+1})$, followed by a length $N_{j+2}$ burst from time $i+T-\sum_{l=j+2}^{L+1} N_l+1$ to $i+T-\sum_{l=j+3}^{L+1} N_l$ introduced on channel $(r_{j+1},r_{j+2})$ and so on, up to a length $N_{L+1}$ burst from time $i-N_{L+1}+1$ to $i+T$ introduced on channel $(r_{L},r_{L+1})$ would result in a decoding failure for node $r_j$ and all the nodes up to the destination $r_{L+1}$.

Further, in addition, for every $i\in\mathbb{Z}_+$, message ${\bf s}_i$ has to be perfectly recovered from
\begin{align}
\left({\bf y}^{(r_{j})}_{i+\sum_{l=1}^{j-1}N_l},{\bf y}^{(r_{j})}_{i+\sum_{l=1}^{j-1}N_l+1},\ldots,{\bf y}^{(r_{j})}_{i+T-\sum_{l=j+1}^{L+1} N_l}\right)
\end{align}
by node $r_j$ given that ${\bf s}_0,{\bf s}_1,\ldots,{\bf s}_{i-1}$ have been correctly decoded by node $r_j$, or otherwise a length $N_1$ burst erasure from time $i$ to $i+N_1-1$ induced on channel $(r_0,r_1)$ followed by a length $N_2$ burst erasure from time $i+N_1$ to $i+N_1+N_2-1$ induced on channel $(r_1,r_2)$ and so on (up to a burst erasure $N_{j-1}$ from time $i+\sum_{l=1}^{j-2}N_l$ to $i+\sum_{l=1}^{j-1}N_l-1$ induced on channel $(r_{j-2}, r_{j-1})$) would result in a decoding failure for node $r_j$. These constraints are depicted in Figure~\ref{fig:segmentJ}.

Recalling that $N_1^{j-1}=\sum_{l=1}^{j-1}N_l$, and $N_{j+1}^{L+1}=\sum_{l=j+1}^{L+1}N_l$ it follows that
\begin{align}
    H\left({\bf s}_i\given[\Big]\left\{{\bf x}_{i+N_1^{j-1}}^{(r_{j-1})},{\bf x}_{i+N_1^{j-1}+1}^{(r_{j-1})},\ldots,{\bf x}_{i+T-N_{j+1}^{L+1}}^{(r_{j-1})} \right\}\setminus \right.  \left. \left\{{\bf x}_{\theta_1},\ldots,{\bf x}_{\theta_{N_j}}\right\},{\bf s}_0,\ldots,{\bf s}_{i-1}\right)=0
    \label{eq:middleBasic}
\end{align}

for any $i\in\mathbb{Z}_+$ and $N_j$ non-negative integers denoted by $\theta_1,\ldots,\theta_{N_j}$. We note that for all $q\in\mathbb{Z}_+$,
\begin{align}
    &\given[\Big] \left\{q+N_1^{j-1},q+N_1^{j-1}+1,\dots,q+T-N_{j+1}^{L+1}\right\}\bigcap \left\{N_1^{j-1}+N_j+m\cdot(T-N_1^{j-1}-N_{j+1}^{L+1}+1),\right. \nonumber \\
    & \left.
    ~~~~N_1^{j-1}+N_j+1+m\cdot(T-N_1^{j-1}+N_{j+1}^{L+1}+1),\ldots,T-N_{j+1}^{L+1}+m\cdot(T-N_1^{j-1}-N_{j+1}^{L+1}+1)\right\}_{m=0}^j \given[\Big] \nonumber \\
    & =T-N_1^{j-1}-N_j-N_{j+1}^{L+1}+1
    \label{eq:countMiddle}
\end{align}

Using (\ref{eq:middleBasic}), (\ref{eq:countMiddle}) and the chain rule, we have
\begin{align}
    &H\left({\bf s}_0,\ldots,{\bf s}_{T-N_1^{j-1}-N_{j+1}^{L+1}+(j-1)(T-N_1^{j-1}-N_{j+1}^{L+1}+1)}\given[\Big] \right.  \nonumber \\
    &\left.\left\{{\bf x}^{(r_{L})}_{N_1^{j-1}+N_{j}+m\cdot(T-N_1^{j-1}-N_{j+1}^{L+1}+1)},{\bf x}^{(r_{L})}_{N_1^{j-1}+N_{j}+1+m\cdot(T-N_1^{j-1}-N_{j+1}^{L+1}+1)},
    \ldots, {\bf x}^{(r_{L})}_{T-N_{j+1}^{L+1}+m\cdot(T-N_1^{j-1}-N_{j+1}^{L+1}+1)}\right\}_{m=0}^j \right) \nonumber \\
    &=0.
    \label{eq:outcomeMid}
\end{align}

Therefore, following the arguments in \cite{fong2018optimal} which we recall here and in Appendix~\ref{sec:app1}, we note that the \eqref{eq:outcomeMid} involves $j(T-N_1^{j-1}-N_{j+1}^{L+1}+1)$ source messages and $(j+1)(T-N_1^{L+1}+1)$ source packets. Therefore, the $(N_1,\ldots,N_{L+1})$-achievable $(n_1,\ldots,n_{L+1}k,T)_{\mathbb{F}}$-streaming code restricted to channel $(r_{j-1},r_{j})$ can be viewed as a point-to-point streaming code with rate $k/n_{j}$ and delay $T-N_1^{j-1}-N_{j+1}^{L+1}$ which can correct any $N_{j}$ erasures. In particular the point-to-point can correct the periodic erasure sequence $\breve{e}^{\infty}$ depicted in Figure~\ref{fig:Periodic3}, which is formally defined as
\begin{align}
    \tilde{e}_i=\begin{cases}
    0~~~{\rm if}~0\leq i\mod(T-N_1^{j-1}-N_{j+1}^{L+1}+1)\leq T-N_1^{L+1} \\
    1~~~{\rm otherwise}
    \end{cases}
\end{align}
\begin{figure}[h]
    \begin{center}
        \begin{psfrags}
            \psfrag{d1}[][][1]{$\ldots$}
            \psfrag{a}[][][1]{$T-N_1^{L+1}+1$}
            \psfrag{b}[][][1]{$N_{j}$}
            \includegraphics[width=0.5\columnwidth]{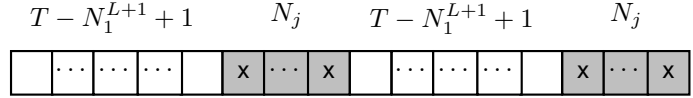}
        \end{psfrags}
    \end{center}
    \caption{A periodic erasure sequence with period $T-N_1^{j-1}-N_{j+1}^{L+1}+1$.}
    \label{fig:Periodic3}
\end{figure}
By standard arguments which are rigorously elaborated in Appendix~\ref{sec:app1}, we conclude that
\begin{align}
    \frac{k}{n_{j}}&\leq \frac{T-\sum_{l=1}^{L+1} N_l+1}{T-\sum_{l=1,l\neq j}^{L+1} N_l +1} \nonumber \\
    &=C_{T-\sum_{l=1,l\neq j}^{L+1} N_l,N_{j}}.
    \label{eq:generalRelay}
\end{align}

Therefore we have
\begin{align}
    R&\leq \frac{k}{\max_{j\in\{1,2,\ldots,L+1\}}\left\{n_j\right\}}\nonumber \\
    &=\frac{T-\sum_{l=1}^{L+1} N_l+1}{T-\min_{j}\left\{\sum_{l=1, l\neq j}^{L+1} N_l\right\}+1} \nonumber \\
    &=\min_{j} C_{T-\sum_{l=1,l\neq j}^{L+1} N_l,N_{j}}\nonumber \\
    & =C^{+}_{T,N_1,\ldots,N_{L+1}}.
\end{align}

\begin{figure}[h]
    \begin{center}
        \begin{psfrags}
            \psfrag{d1}[][][1.5]{$\ldots$}
            \psfrag{x}[][][1]{${\rm x}$}
            \psfrag{b1}[][][\fontSizeIfFig]{$\underbrace{~~~~~~~~~~~}_{N_1}$}
            \psfrag{b2}[][][\fontSizeIfFig]{$\underbrace{~~~~~~~~~~~}_{N_{j-1}}$}
            \psfrag{b3}[][][\fontSizeIfFig]{$\underbrace{~~~~~~~~~~~}_{N_{j+1}}$}
            \psfrag{b4}[][][\fontSizeIfFig]{$\underbrace{~~~~~~~~~~~~~~}_{N_{L+1}}$}
            \psfrag{c}[][][1.5]{$\ddots$}
            \psfrag{a1}[][][\fontSizeIfFig]{\fbox{$i$}}
            \psfrag{a2}[][][\fontSizeIfFig]{\fbox{$i+N_1-1$}}
            \psfrag{a3}[][][\fontSizeIfFig]{\fbox{$i+N_{1}^{j-2}$}}
            \psfrag{a4}[][][\fontSizeIfFig]{\fcolorbox{black}{gray!15}{$i+N_{1}^{j-1}-1$}}
            \psfrag{a5}[][][\fontSizeIfFig]{\fbox{$i+N_{1}^{j-1}$}}
            \psfrag{a10}[][][\fontSizeIfFig]{\fbox{$i+T$}}
            \psfrag{a9}[][][\fontSizeIfFig]{\fbox{$i+T-N_{L+1}+1$}}
            \psfrag{a8}[][][\fontSizeIfFig]{\fbox{$i+T-N_{j+2}^{L+1}$}}
            \psfrag{a7}[][][\fontSizeIfFig]{\fcolorbox{black}{gray!15}{$i+T-N_{j+1}^{L+1}+1$}}
            \psfrag{a6}[][][\fontSizeIfFig]{\fbox{$i+T-N_{j+1}^{L+1}$}}
            \psfrag{r1}[][][\fontSizeIfFig]{$\rm{Link~}(r_0,r_1)$}
            \psfrag{r2}[][][\fontSizeIfFig]{$\rm{Link~}(r_{j-2},r_{j-1})$}
            \psfrag{r3}[][][\fontSizeIfFig]{$\rm{Link~}(r_{j-1},r_{j})$}
            \psfrag{r4}[][][\fontSizeIfFig]{$\rm{Link~}(r_{j},r_{j+1})$}
            \psfrag{r5}[][][\fontSizeIfFig]{$\rm{Link~}(r_L,r_{L+1})$}
            \includegraphics[width=\twoColFigSize\columnwidth]{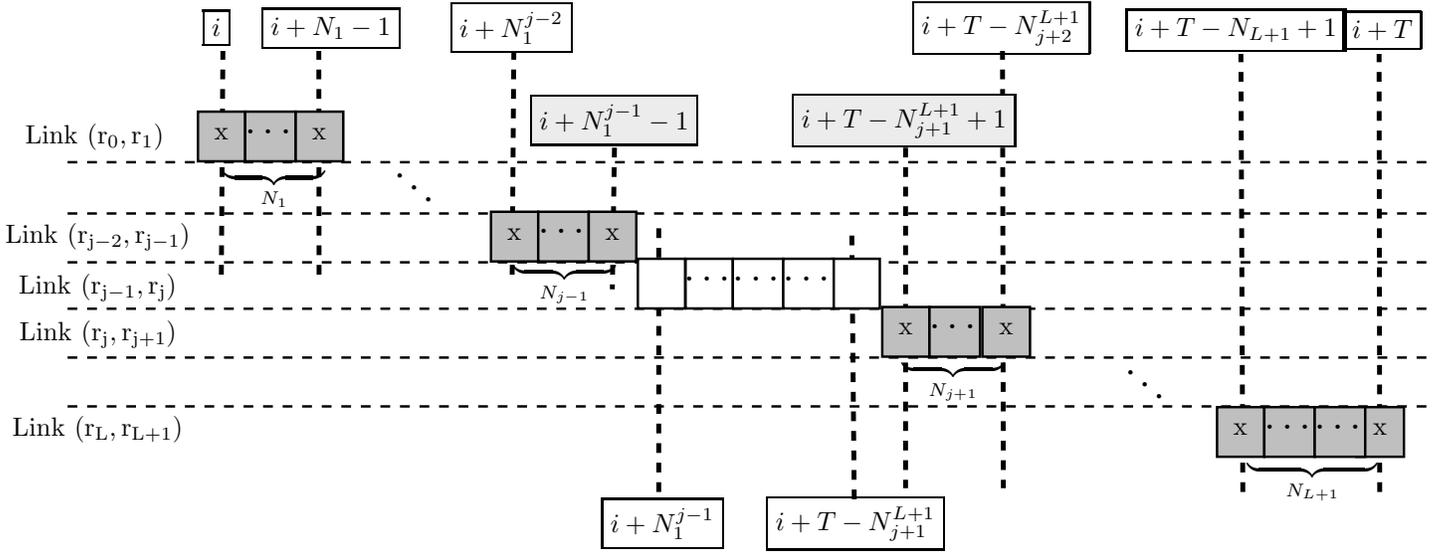}
        \end{psfrags}
    \end{center}
    \caption{Constrains imposed on transmission in link $(r_{j-1},r_{j})$.}
    \label{fig:segmentJ}
\end{figure}

\section{State-dependent symbol-wise DF}
\label{sec:codScheme}
As mentioned above, the achievable scheme we analyze is a state-dependent symbol-wise DF scheme. This scheme is composed from a block code combined with diagonal interleaving. We start with some basic definitions of point-to-point block codes which would be the basis for this scheme.
\begin{definition}
A point-to-point $(n,k,T)_{\mathbb{F}}$-block code consists of the following
\begin{enumerate}
    \item A sequence of $k$ symbols $\{u[l]\}_{l=0}^{k-1}$ where $u[l]\in\mathbb{F}$.
    \item A generator matrix $\svv{G}\in\mathbb{F}^{k\times n}$. The source codeword is generated according to
    \begin{align}
        \left[x[0]~x[1]~\ldots~x[n-1]\right]=\left[u[0]~u[1]~\ldots~u[k-1]\right]\svv{G}
    \end{align}
    \item A decoding function $\varphi_{l+T}:\mathbb{F}\cup\{*\}\times\ldots\mathbb{F}\cup\{*\}\to\mathbb{F}$ for each $l\in\{0,1,\ldots,k-1\}$, where $\varphi_{l+T}$ is used by the destination at time $\min\{l+T,n-1\}$ to estimate $u[l]$ according to
    \begin{align}
        \hat{u}[l]=\begin{cases}
        \varphi_{l+T}(y[0],y[1],\ldots,y[l+T])~~~{\rm if}~l+T\leq n-1 \\
        \varphi_{l+T}(y[0],y[1],\ldots,y[n-1])~~~{\rm if}~l+T> n-1
        \end{cases}
    \end{align}
\end{enumerate}
\end{definition}

\begin{definition}
A point-to-point $(n, k, T)_{\mathbb{F}}$-block code is said to be $N$-achievable if the following holds for any
$N$-erasure sequence $e^{\infty}\in \Omega_N$: For the $(n, k, T)_{\mathbb{F}}$-block code, we have
\begin{align}
    \hat{u}[l]=u[l]
\end{align}
for all $l\in\{0,1,\ldots,k-1\}$ and $u[l]\in\mathbb{F}$ and all $u[l]\in\mathbb{F}$, where
 \begin{align}
        \hat{u}[l]=\begin{cases}
        \varphi_{l+T}(g_1(x[0],e_0),\ldots,g_1(x[l+T],e_{l+T}))~~~{\rm if}~l+T\leq n-1 \\
        \varphi_{l+T}(g_1(x[0],e_0),\ldots,g_1(x[n-1],e_{n-1}))~~~{\rm if}~l+T> n-1
        \end{cases}
    \end{align}
with $g_1$ being the symbol-wise erasure function defined in \eqref{eq:g}.
\end{definition}

Recalling from Section~\ref{sec:converse} that the rate in channel $(r_{j},r_{j+1})$ is upper bounded by $R_{j}\leq\frac{k}{n_{j}}$ where
\begin{align}
    k&=T-\sum_{l=1}^{L+1}N_l+1 \nonumber \\
    n_{j+1}&=T-\sum_{l=1,l\neq {j+1}}^{L+1}N_l+1.
    \label{eq:k_nj}
\end{align}

Denoting 
\begin{align}
    T_{j+1}=T-\sum_{l=1,i\neq {j+1}}N_l,
    \label{eq:T_j}
\end{align}

The suggested coding scheme is composed of $(n_{j+1},k,N_1^{j}+T_j)_{\mathbb{F}}$-block codes combined with diagonal interleaving. Each relay $r_j$ is using $n_{j+1}$ codes, each one is an $(n_{j+1}, k,N_1^{j}+T_j)_{\mathbb{F}}$-block code that depends on the erasure pattern in the previous relay, i.e., each code can be different.

More formally, let $s_i[l]$ be the $l$'th symbol of the source message ${\bf s}_i$ and let $x^{(r_j)}_i[l]$ be the $l$'th symbol of the output of encoding function $f_i^{(r_j)}$. We may say for each $i\in\mathbb{Z}_+$, a single transmission function of relay $r_j$ constructs
\begin{align}
    \left[x^{(r_j)}_{i+N_1^j}[0]~x^{(r_j)}_{i+N_1^j+1}[1]~\cdots~x^{(r_j)}_{i+N_1^j+n_{j+1}-1}[n_{j+1}-1]\right] \triangleq\left[s_i[0]~s_{i+1}[1]~\cdots~s_{i+k-1}[k-1]\right]\times \svv{G}_i^{(r_j)},
    \label{eq:diagInter}
\end{align}
where we show next that $\svv{G}_i^{(r_j)}$ is a $k\times n_{j+1}$ generator matrix of a $(n_{j+1},k)$ MDS code. We assume that for any $i<0$, ${\bf s}_i={\bf 0}$.

Denoting with 
\begin{align}
    {\bf \tilde{s}}_i=\left[s_i[0]~s_{i+1}[1]~\cdots~s_{i+k-1}[k-1]\right],
    \label{eq:tilde_s}
\end{align}
an example of the diagonal interleaving of a single code is given in Table~\ref{table:SingleDiagRj} below.
\begin{table}[h]
\begin{center}
    \begin{tabular}{|c|c|c|c|c|}
        \hline 
        Time & $i+N_1^j$ & $i+N_1^j+1$ & $\cdots$ & $i+N_1^j+n_{j+1}-1$ \\
        \hline 
        & $[{\bf \tilde{s}}_i\times \svv{G}_i^{(r_j)}][0]$ & & & \\ \hline 
        & & $[{\bf \tilde{s}}_i\times \svv{G}_i^{(r_j)}][1]$ & & \\
         \hline
        & & & $\ddots$ & \\
         \hline
        & & & & $[{\bf \tilde{s}}_i\times \svv{G}_i^{(r_{j+1)}}][n_{j+1}-1]$ \\ \hline
    \end{tabular}
\end{center}
\caption{An example of a single code transmitted in $r_j$}
\label{table:SingleDiagRj}
\end{table}

Therefore, each transmitted packet at relay $r_j$ is composed of $n_{j+1}$ symbols, each of which is taken from a different code. Recalling that $[{\bf \tilde{s}}_i\times \svv{G}_i^{(r_j)}][j]$ means that we take the $j$'th element from the vector resulting from multiplying ${\bf \tilde{s}}_i$ with the generator matrix $\svv{G}_i^{(r_j)}$, an  example of a packet sent by relay $r_j$  is given in \eqref{eq:xSentAtRj} below.

\begin{equation}
    {\bf x}_{i+N_1^j}^{(r_j)}=
    \begin{bmatrix}
    x_{i+N_1^j}^{(r_j)}[0] \\
    x_{i+N_1^j}^{(r_j)}[1] \\
    \vdots \\
    x_{i+N_1^j}^{(r_j)}[n_{j+1}-1] \\
    \end{bmatrix} = 
    \begin{bmatrix}
    [{\bf \tilde{s}}_i\times \svv{G}_i^{(r_j)}][0] \\
    [{\bf \tilde{s}}_{i-1}\times \svv{G}_{i-1}^{(r_j)}][1] \\
    \vdots \\
    [{\bf \tilde{s}}_{i-n_{j+1}+1}\times \svv{G}_{i-n_{j+1}+1}^{(r_j)}][n_{j+1}-1] \\
    \end{bmatrix}
    \label{eq:xSentAtRj}
\end{equation}

We note again that the specific structure of each $\svv{G}_i^{(r_j)}$ is defined according to the erasure pattern of the previous links.


We describe next the process of generating $\svv{G}_i^{(r_j)}$ (required to transmit ${\bf{\tilde{s}}}_i$ in each one of the relays).

\begin{itemize}
    \item At the sender ($r_0$), use an $(n_1,k)$ MDS code. Hence, $\svv{G}_i^{(r_0)}$ is the generator matrix of an $(n_1,k)$ MDS code.
    \item Each encoder at relay $r_j$ ($j\in\{1,\ldots,L\}$) performs the following
    \begin{enumerate}
        \item Store any non-erased symbols from the first $N_{j}$ received symbols from link $(r_{j-1},r_{j})$, i.e., all non-erased symbols from  
        \begin{align}
            \left\{x_{i+{N_1}^{j-1}}^{(r_{j-1})}[0],\ldots,x_{i+{N_1}^{j-1}+N_j-1}^{(r_{j-1})}[N_j-1]\right\}.
        \end{align}
        \item Start transmitting at time $i+N_1^j$ (while continuing to store the received symbols from link $(r_{j-1},r_{j})$). Until time $i+N_1^j+k-2$, forward the $k-1$ symbols received from link $(r_{j-1},r_{j})$ by the order they were received, i.e., forward the $k-1$ non-erased symbols from  
        \begin{align}
            \left\{x_{i+{N_1}^{j-1}}^{(r_{j-1})}[0],\ldots,x_{i+{N_1}^{j-1}+N_j+k-2}^{(r_{j-1})}[N_j+k-2]\right\}.
        \end{align}
        Noting the $N_j+k-2=n_j-1$ means that relay forwards the $k-1$ symbols received by the order they were received, i.e., all non-erased symbols from  
        \begin{align}
            \left\{x_{i+{N_1}^{j-1}}^{(r_{j-1})}[0],\ldots,x_{i+{N_1}^{j-1}+N_j-1}^{(r_{j-1})}[n_{j}-1]\right\}.
        \end{align}
        \item At time $N_1^j+k-1$, recover ${\bf {\tilde{s}}_i}$. In Lemma~\ref{lem:lem1} below we prove that it is feasible for any $N_1,\ldots,N_{L+1}$-erasure sequence.
        \item Transmit until time $N_1^j+n_{j+1}-1$ encoded symbols. The encoded symbols should be non-received symbols from $(n_{max},k)$ MDS code to be defined below.
        \item For each transmitted symbol, attach a header which will be defined next.
    \end{enumerate}
\end{itemize}



Recalling the definition of $n_{\max}$ \eqref{eq:Tj}, the following Proposition sheds light on the suggested method of encoding $\bf{\tilde{s}}_i$ once it is recovered at $r_j$ (step (4)) in order to guarantee that each $\svv{G}^{(r_j)}_i$ is a generator matrix of $(n_{j+1},k)$ MDS code and further simplify the header required to allow decoding.


\begin{proposition}
All block codes used by nodes $r_j$ where $j\in\{0,\ldots,L\}$ to transmit ${\bf {\tilde{S}}}_i$ can be generated by puncturing and applying a permutation to the $(n_{\max},k)$ MDS code which is associated with rate $C^{+}_{T,N_1,\ldots,N_{L+1}}$.
\label{prop:prop1}
\end{proposition}
This proposition holds since $k$ (the number of information symbols) is the same for all codes. Denoting $N_{\max}=\max_j{N_j}$, we recall that the MDS code $(n_{\max},k)$ which is the MDS code associated with rate $C^{+}_{T,N_1,\ldots,N_{L+1}}$ can correct any $N_{max}$. Puncturing any $N_{\max}-N_j$ columns from the generator matrix of this code results with a code that can correct any $N_j$ erasures.

In fact, denoting with $\svv{G}_{\max}$ the generator matrix of $(n_{\max},k)$ MDS code, $\svv{G}_i^{(r_j)}$ can be viewed as taking $n_{j+1}$ columns from $\svv{G}_{\max}$ and apply permutation on the order of the columns. The specific columns taken from $\svv{G}_{\max}$ and their order is defined by the specific erasure pattern which occurs. 

Further, following Proposition~\ref{prop:prop1} we define the header as a number which indicates the location of the column from $\svv{G}_{\max}$ that was used to generate this symbol. Thus, the header attached to each symbol transmitted from each relay is a number in the range $\left[1,\ldots,n_{\max}\right]$.

Thus we have the following Lemma.
\begin{lemma}
For any $N_1,\ldots,N_{L+1}$-erasure sequence, any $j\in\{0,\ldots,L\}$ and any $i\in \mathbb{Z}_+$, $\svv{G}_i^{(r_j)}$ is a generator matrix of $(n_{j+1},k)$ MDS code. 
\label{lem:lem1}
\end{lemma}

\begin{proof}
From the construction, $\svv{G}_i^{(r_0)}$ is an $(n_1,k)$ MDS code. Hence we assume by induction that $\svv{G}_i^{(r_{j-1})}$ is an $(n_{j},k)$ MDS code and show that $\svv{G}_i^{(r_{j})}$ is an $(n_{j+1},k)$ MDS code.

We note that the only non-trivial steps in generating $\svv{G}_i^{(r_j)}$ are steps (2) and (3). Assuming $N_1,\ldots,N_{L+1}$-erasure sequence means that for any $j\in\{0,\ldots,L\}$ 
\begin{align}
    e^{j,\infty}\in\Omega_{N_{j}}
\end{align}
i.e., that the maximal number of erasures in line $(r_{j-1},r_j)$ is $N_j$. Since we assumed $\svv{G}_i^{(r_{j-1})}$ is an $(n_{j},k)$ MDS code it follows that it is guaranteed that $k$ symbols out of the $n_j$ transmitted symbols will not be erased. Further, since relay $r_{j-1}$ starts sending the coded symbols at time $i+N_1^{j-1}$ and relay $r_j$ starts forwarding the non-erased symbols received from $r_{j-1}$ at time $i+N_1^{j-1}+N_j$ (after buffering any non erased symbols from the first $N_j$ coded symbols) it is guaranteed that relay $r_j$ could forward the $k-1$ non-erased coded symbols sent from $r_{j-1}$.

In step (3), relay $r_j$ needs to recover all $k$ information symbols at time $i+N_1^j+k-1$. We note that this step is feasible since, assuming $\svv{G}_i^{(r_{j-1})}$ is the generator matrix of an $(n_{j},k)$ MDS code, any of its $k$ information symbols can be recovered from any $k$ non-erased symbols. We note that relay $r_{j-1}$ transmits its code in time indices
\begin{align}
    [i+N_1^{j-1},\ldots,i+N_1^{j-1}+n_{j}-1]
\end{align}
therefore, the last symbol of this code is received at relay $r_j$ at index
\begin{align}
    i+N_1^{j-1}+n_{j}&=i+N_1^{j-1}+T-\sum_{l=1,l\neq j}^{L+1}N_l+1 \nonumber \\
    & = i+N_1^{j-1}+N_j+(T-\sum_{l=1}^{L+1}N_l+1)  \nonumber \\
    & = i+N_1^{j}+k.
\end{align}

These regions are depicted in Figure~\ref{fig:transScheme} below.
\begin{figure}[h]
    \begin{center}
        \begin{psfrags}
            \psfrag{d1}[][][1.5]{$\ldots$}
            \psfrag{b1}[][][1]{$n_{j+1}$}
            \psfrag{b2}[][][1]{$n_{j}$}
            \psfrag{b3}[][][1]{$n_{j}-N_j=k$}
            \psfrag{a1}[][][1]{\fbox{$i+N_1^{j}$}}
            \psfrag{a2}[][][1]{\fbox{$i+N_1^{j-1}+n_j=i+N_1^{j}+k$}}
            \psfrag{a3}[][][1]{\fbox{$i+N_1^{j-1}$}}
            \psfrag{a4}[][][1]{\fbox{$i+N_1^{j}+n_{j+1}$}}
            \psfrag{c1}[][][1]{$(r_{j-1},r_j)$}
            \psfrag{c2}[][][1]{$(r_{j},r_{j+1})$}
            \includegraphics[width=0.5\columnwidth]{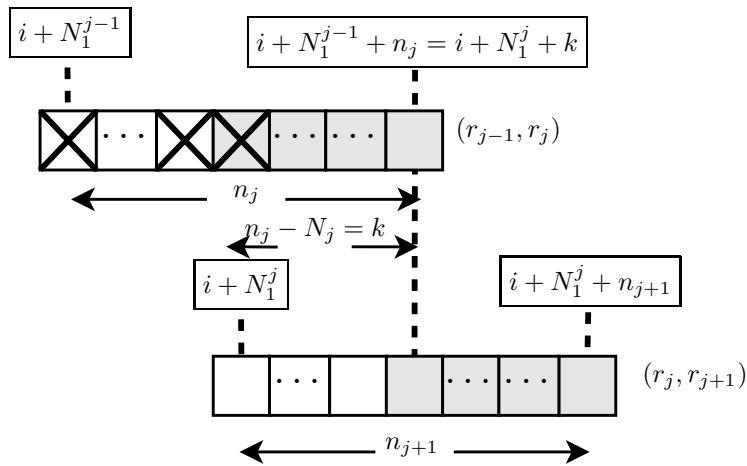}
        \end{psfrags}
    \end{center}
    \caption{Two regions of transmission in link $(r_j,r_{j+1})$. The symbols with white background are symbols forwarded from link $(r_{j-1},r_{j})$. The shaded symbols are transmitted after the $k$ information symbols are decoded hence they are either symbols erased in link $(r_{j-1},r_{j})$ or additional (independent) linear combinations of the information symbols.}
    \label{fig:transScheme}
\end{figure}

Since we assumed $\svv{G}_i^{(r_{j-1})}$ is an $(n_{j},k)$ MDS code, it follows its $k$ information symbols can be recovered from any $k$ symbols received. However, in order to recover the information symbols, relay $r_j$ needs to know the structure of $\svv{G}^{(r_{j-1})}$. Using the header attached to each received symbol relay $r_j$ it is guaranteed that relay $r_j$ could recover (all $k$) information symbols at time $i+N_1^j+k-1$. 

Further, since after the recovery of the $k$ symbols, $r_j$ adds $n_{j+1}-k$ unique columns from $\svv{G}_{\max}$ the process described above results with $\svv{G}^{(r_{j})}$ which is the generator matrix of $(n_{j+1},k)$ MDS code which can recover any $N_{j+1}$ erasures. 


\end{proof}

For specific examples on the modification of the block codes as a function of rate change between links, see Appendix~\ref{app:rateChangeInRelay}.

Next we show the following lemma.
\begin{lemma}[Based on Lemma 3 in \cite{fong2018optimal}]
Suppose $T\geq N$, and let $k\triangleq T-N+1$ and $n\triangleq k+N$. For any $\mathbb{F}$ such that $\mathbb{F}\geq n$, there exists an $N$-achievable point-to-point $(n,k,T)_{\mathbb{F}}$-block code.
\end{lemma}
\begin{proof}
The proof follows directly from the definitions of MDS code. Any $(n,k)$ MDS code is $(n,k,n-1)$ block code. Thus, any $(n,k)$ MDS code is $(n-k)$-achievable where all symbols can be decoded by the end of the code block.
\end{proof}

Recalling that when transmitting ${\bf {\tilde s}}_i$, relay $r_j$ starts its transmission at time $i+N_1^{j}$ we have the following corollary.
\begin{corollary}
Recalling that relay $r_j$ (for any $j\in\{0,\ldots,L\}$) starts transmitting the coded symbols of ${\bf\tilde{s}}_i$ at time $i+N_1^{j}$, it follows that for any $N_1,\ldots,N_{L+1}$-erasure sequence the code used in each relay $r_j$  to transmit ${\bf {\tilde s}}_i$ is $N_{j+1}$-achievable $(n_{j+1},k,N_1^{j}+T_{j+1})_{\mathbb{F}}$ point-to-point block code, i.e., all the symbols of ${\bf {\tilde s}}_i$ can be decoded at relay $r_{j+1}$ by delay of 
\begin{align}
    i+N_1^{j}+T_{j+1}&=i+N_1^{j}+T-\sum_{l=1,l\neq j+1}^{L+1}N_l \nonumber \\
    & = i+T-\sum_{l=j+2}^{L+1}N_l.
\end{align}
\label{col:col2}
\end{corollary}
We take a closer look at the channel between the last relay and the destination $(r_{L},r_{L+1})$ . Following Corollary~\ref{col:col2} we have
\begin{corollary}
For any $i\in\mathbb{Z}_{+}$, and for any $N_1,\ldots,N_{L+1}$-erasure sequence, ${\bf {\tilde s}}_i$ can be decoded at the destination at time 
\begin{align}
     i+T-\sum_{l=L+2}^{L+1}N_l=i+T,
\end{align}
at the destination, i.e., using the construction suggested above, ${\bf {\tilde s}}_i$ can be decoded at the destination with delay of $T$ for any $e^{L+1,\infty}\in\Omega_{N_{L+1}}$.
\label{col:col3}
\end{corollary}
Next, we show that this Corollary means that the construction suggested above generates a $(n_1,\ldots,n_{L+1},k,T)_{\mathbb{F}}$ streaming code which is also $N_1,\ldots,N_{L+1}$-achievable.

\begin{lemma}
The streaming code resulting from using $\svv{G}_i^{(r_j)}$ defined above in each relay $j\in\{0,\ldots,L\}$ for every $i\in\mathbb{Z}_{+}$ is a $(n_1,\ldots,n_{L+1},k,T)_{\mathbb{F}}$ streaming code which is also $N_1,\ldots,N_{L+1}$-achievable.
\label{lem:lem3}
\end{lemma}
\begin{proof}
Our goal is to show that the destination ($r_{L+1)}$) can recover ${\bf s}_i=\left[s_i[0],s_i[1],\cdots,s_i[k]\right]$ based on
\begin{align}
    \left[{\bf y}_0^{(r_{L+1})},{\bf y}_1^{(r_{L+1})},\cdots,{\bf y}_{i+T}^{(r_{L+1})}\right]=\left[g_{n_{L+1}}\left({\bf x}^{(r_L)}_0,e^{L+1}_0\right),\ldots,g_{n_{L+1}}\left({\bf x}^{(r_L)}_{i+T},e^{L+1}_{i+T}\right)\right]
\end{align}

Following Lemma~\ref{lem:lem1} it follows each $\svv{G}_i^{(r_j)}$ is $n_{j+1},k$ MDS code. From Corollary~\ref{col:col3} it follows that for any $i\in\mathbb{Z}_+$, $s_i[0]$ (which is the first element in ${\bf {\tilde{s}}}_i$) can be recovered with a delay of $T$ for any $N_1,\ldots,N_{L+1}$-erasure sequence. Similarly we note that $s_i[1]$ can be recovered with a delay of $T-1$ and $s_i[k-1]$ can be recovered with delay of $T-k$. Thus, we conclude that the $(n_1,\ldots,n_{L+1},k,T)_{\mathbb{F}}$ streaming code resulting from the construction described above is a $(n_1,\ldots,n_{L+1},k,T)_{\mathbb{F}}$ streaming code which is also $N_1,\ldots,N_{L+1}$-achievable.
\end{proof}


Thus to prove Theorem~\ref{thm:thm3}, we need to show that when $|\mathbb{F}|\geq n_{\max}$, $\svv{G}^{(r_j)}_i$ can be generated at each relay $r_j$ for $j\in\{0,\ldots,L\}$ and analyze the overall rate (by bounding the rate of the additional header).

\begin{proof}[Proof of Theorem \ref{thm:thm3}]
We first note that, as mentioned in Section~\ref{sec:networkModel}, an $(n_{max},k)$ MDS code exists as long as $|\mathbb{F}|\geq {n_{\max}}$. Therefore, following Lemma~\ref{lem:lem1}, when $|\mathbb{F}|\geq {n_{\max}}$ it follows that for any $N_1,\ldots,N_{L+1}$-erasure sequence, any $j\in\{0,\ldots,L\}$ and any $i\in \mathbb{Z}_+$, there exists $\svv{G}_i^{(r_j)}$ which is a generator matrix of $(n_{j+1},k)$ MDS code. 

Following Lemma~\ref{lem:lem3} it follows that streaming code resulting from using $\svv{G}_i^{(r_j)}$ defined above in each relay $j\in\{0,\ldots,L\}$ for every $i\in\mathbb{Z}_{+}$ is a $(n_1,\ldots,n_{L+1},k,T)_{\mathbb{F}}$ streaming code which is also $N_1,\ldots,N_{L+1}$-achievable. The rate in relay $r_j$, without taking the size of the header into account, is $\frac{k}{n_{j+1}}$. Thus, from Definition~\ref{def:def5}, the overall rate of transmission is upper bounded by
\begin{align}
    R&\leq \min_{j\in\{0,1,\ldots,L\}} \frac{k}{n_{j+1}} \nonumber \\
    &= C^{+}_{T,N_1,\ldots,N_L}.
\end{align}

The header attached to each packet sent from relay $r_j$ is composed from stacking the $n_{j+1}$ headers used by each symbol generated from a $(n_{j+1},k,N_1^j+T_j)_{\mathbb{F}}$ block code which is part of each transmission packet. As we defined above, this header is a number from $\left[1,\ldots,n_{\max}\right]$. Hence the size of the header is $n_{j+1}\log(n_{\max})$ bits. We further note that the size of the header is upper bounded by $n_{\max}\log(n_{\max})$.

To conclude, each node transmits $n_{j+1}$ coded symbols (each taken from field $\mathbb{F}$) along with $n_{j+1}\log(n_{j+1})$ bits of header to transfer $k$ information symbols (each taken from field $\mathbb{F}$). The overall rate is
\begin{align}
    R &\geq \min_{j}  \frac{k\cdot\log(|\mathbb{F}|)}{n_{j+1}\cdot\log(|\mathbb{F}|)+n_{\max}\lceil\log\left(n_{\max}\right)\rceil} \nonumber \\
    & = \frac{T-\sum_{l=1}^{L+1}N_l+1}{\max_j\left\{T-\sum_{l=1,l\neq j}^{L+1}N_l+1\right\}+\frac{n_{\max}\lceil\log\left(n_{\max}\right)\rceil}{\log(|\mathbb{F}|)}} \nonumber \\
     & = \frac{T-\sum_{l=1}^{L+1}N_l+1}{T-\min_j\left\{\sum_{l=1,l\neq j}^{L+1}N_l\right\}+1+\frac{n_{\max}\lceil\log\left(n_{\max}\right)\rceil}{\log(|\mathbb{F}|)}}
\end{align}
where $n_{\max}$ is defined in (\ref{eq:Tj}).
\end{proof}

\section{An upper bound on loss probability attained by state-dependent symbol-wise DF for random erasure}
\label{sec:upperBound}

In Section~\ref{sec:codScheme}, the state-dependent symbol-wise DF scheme was described, and a lower bound on its achievable rate was derived while assuming a deterministic erasure model. In this section, we develop an upper bound on the average loss probability when this scheme is applied over channels with random (i.i.d.) erasures.

Let $s_i[0],s_i[1],\dots,s_i[k-1]$ be the k source symbols transmitted by node $r_0$ at each discrete time $i$. We note that for the $n_1$, $(n_1,k)$ MDS codes used by the sender, the following property holds:
\begin{itemize}
    \item For every $s_i[v]$ located at the ($v+1$)th position of the length-$k$ packet transmitted at time $i$ by the $(n_1,\ldots,n_{L+1,}k,T)_{\mathbb{F}}$ streaming code over $(r_0,r_1)$, $\hat{s}^{(r_1)}_i[v]$ is generated by the relay, at the latest, at time $i-v+n_1-1$ (i.e., after transmission of $n_1$ symbols from $r_0$). If there are at most $N_1$ erasures inside the window $\{i-v, i-v + 1,\ldots,i-v+n_1-1\}$, then $\hat{s}^{(r_1)}_i[v] = {s}_i[v]$.
\end{itemize}
Hence, ${\bf s}_i$ can be fully recovered at relay $r_1$ if for all $v\in\{0,1,\ldots,k-1\}$, in any window $\{i-v, i-v+1,\ldots,i-v+n_1-1\}$, there are at most $N_1$ erasures. We bound the loss probability by analyzing the probability in which in the window $\{i-k+1,i-k+2,\ldots,i+n_1-1\}$ there are at most $N_1$ erasures.

Since the state-dependent symbol-wise DF encode the same information symbols (per diagonal) in each relay, we note that in the general case, when transmitting the $(n_1,\ldots,n_{L+1,}k,T)_{\mathbb{F}}$ streaming code over $(r_{j-1},r_j)$:
\begin{itemize}
    \item $\hat{\bf s}^{(r_j)}_i[v]$ is generated by the relay, at the latest, at time $i+N_1^{j-1}-v+n_{j}-1$ (i.e., after transmission of $n_{j}$ symbols from relay $r_{j-1}$). If there are at most $N_j$ erasures inside the window $\{i+N_1^{j-1}-v, i-v + 1,\ldots,i+N_1^{j-1}-v+n_j-1\}$, $\hat{s}^{(r_j)}_i[v] = s_i[v]$.
\end{itemize}
Hence, ${\bf s}_i$ can be fully recovered at relay $r_j$ if for all $v\in\{0,1,\ldots,k-1\}$, in any window $\{i+N_1^{j-1}-v, i+N_1^{j-1}-v + 1,\ldots,i+N_1^{j-1}-v+n_j-1\}$ there are at most $N_j$ erasures. Similar to \cite{fong2018optimal}, we bound the loss probability by analyzing the probability in which in the window $\{i+N_1^{j-1}-k+1, i+N_1^{j-1}-k + 1,\ldots,i+N_1^{j-1}+n_j\}$ there are at most $N_j$ erasures.

Denoting the average Loss probability as
\begin{align}
    P_{T,N_1,N_2,\dots,N_{L+1}}\triangleq \lim_{M\to \infty} \frac{1}{M} \prob\{\hat{\bf s}_i\neq {\bf s}_i \}
\end{align}

achieved by the above state-dependent symbol-wise DF strategy under the random erasure model. Define $\alpha_j=\prob(e^j_0=1)$ to be the erasure probability in link $(r_{j-1},r_j)$. According to the achievability conditions we have
\begin{align}
    \prob\left(\hat{\bf s}_i\neq {\bf s}_i \given[\Big] \sum_{u=i-k+1}^{i+n_1+1}e^1_u\leq N_1,\sum_{u=i-k+1}^{i+n_2+1}e^2_u\leq N_2,\ldots,\sum_{u=i-k+1}^{i+n_{j+1}+1}e^{L+1}_u\leq N_{L+1} \right)=0
\end{align}
for every $i\geq T-N_1^{L+1}$. Since
\begin{align}
& \prob\left(\left\{\sum_{u=i-k+1}^{i+n_1+1}e^1_u > N_1 \right\} \bigcup \left\{\sum_{u=i-k+1}^{i+n_2+1}e^2_u> N_2 \right\} \bigcup \ldots \bigcup \left\{\sum_{u=i-k+1}^{i+n_{L+1}+1}e^{L+1}_u > N_{L+1} \right\} \right) \nonumber \\
&\leq \sum_{u=N_{1}+1}^{2k+2N_1+1}{2k+2N_1+1 \choose u} (\alpha_1)^u(1-\alpha_1)^{2k+2N_1+1-u} \nonumber \\
&~~~~ + \sum_{u=N_{2}+1}^{2k+2N_2+1}{2k+2N_2+1 \choose u} (\alpha_2)^u(1-\alpha_2)^{2k+2N_2+1-u}+\ldots \nonumber \\
&~~~~ + \sum_{u=N_{L+1}+1}^{2k+N_{L+1}+1}{2k+2N_{L+1}+1 \choose u} (\alpha_{L+1})^u(1-\alpha_{L+1})^{2k+2N_{L+1}+1-u}
\end{align}

It follows that
\begin{align}
    P_{T,N_1,N_2,\dots,N_{L+1}}&\leq \kappa_1(T,N_1,\ldots,N_{L+1})\cdot (\alpha_1)^{N_1+1}+\kappa_2(T,N_1,\ldots,N_{L+1})\cdot (\alpha_2)^{N_2+1}+\ldots \nonumber \\
    &+\kappa_{L+1}(T,N_1,\ldots,N_{L+1})\cdot (\alpha_{L+1})^{N_{L+1}+1}
\end{align}
where $\kappa_j(T,N_1,\ldots,N_{L+1})$ does not depend on $\alpha_j$ (or on any other $\alpha_k$ for any $k\neq j$). Hence, $P_{T,N_1,N_2,\dots,N_{L+1}}$ decays exponentially fast in $\min\{N_1+1,N_2+1,\ldots,N_{L+1}+1\}$.

\section{Numerical results}
\label{sec:numerical}

In this section, we show the performance of the state-dependent symbol-wise DF scheme on a random model. We consider a statistical $L+1$-node relay network where i.i.d. erasures are independently introduced to all channels. We denote with $\alpha_j$ the probability of experiencing an erasure in each time slot for channel $(r_{j-1},r_j)$.

Similar to \cite{fong2018optimal}, we will compare state-dependent symbol-wise DF with message-wise DF and instantaneous forwarding, which we briefly recall. In message-wise DF, all the symbols in the same source message are decoded by relay $r_j$ subject to the delay constraint $T_j$ such that $\sum_j T_j\leq T$. The overall rate of message-wise DF is
\begin{align}
    R^{\rm Message}_{T,N_1,N_2,\ldots,N_{L+1}}=\max_{(T_1,\ldots,T_{L+1}):\sum_j T_j\leq T}\min\{C_{T_1,N_1},C_{T_2,N_2},\ldots,C_{T_{L+1},N_{L+1}}\}.
\end{align}
More precisely, we consider message-wise DF scheme constructed by concatenating $L+1$ streaming codes where the $j$'th code is a $(n^{\rm Message}_j, k, T_j)_\mathbb{F}$-streaming code.

We also consider an instantaneous forwarding (IF) strategy, which uses a point-to-point streaming code over the $L+1$-node relay network as if the network is a point-to-point channel. More specifically, under the IF strategy, the source transmits symbols generated by the streaming code and relay $r_j$
forwards every symbol received from relay $r_{j-1}$ in each time slot. The overall point-to-point channel induced by the IF strategy experiences an erasure if either one of the channels experiences an erasure. This results with rate
\begin{align}
    R^{\rm IF}=C_{T,\sum_l N_l}.
\end{align}

We first study the error-correcting capabilities of all schemes in case of having two relays with a rate equal to 2/3 and maximal total delay of $T=9$. We further simulate a symmetric topology, i.e., we assume the same error probability for all segments. Since we assume symmetric topology, we focus on schemes that have the same error-correcting capabilities for all segments.
\begin{itemize}
    \item State-dependent symbol-wise DF can support $(N_1,N_2,N_3)=2$ erasures in each segment which is constructed using concatenating three $(6,4,5)_{\mathbb{F}}$ streaming code. While the rate of the code is strictly lower than $2/3$ due to the overhead it uses, as we noted above, it approaches $2/3$ as $|\mathbb{F}|$ increases.
    \item Message-wise DF can support $(N_1,N_2,N_3)=1$ erasures in each segment which is constructed using concatenating three $(3,2,2)_{\mathbb{F}}$ streaming code. As mentioned in \cite{fong2019optimal}, higher rate codes (such as $(4,3,3)_{\mathbb{F}}$) are excluded since $(3,2,2)_{\mathbb{F}}$ can correct more erasure patterns.
    \item For IF, we note that since $C_{9,\sum_j N_j}=\frac{10-\sum_j N_j}{10}$ and we require $C_{9,\sum_j N_j}\geq 2/3$ we get $\sum_j N_j\leq 3$. Hence IF uses a $(10,7,9)_{\mathbb{F}}$ streaming code.
\end{itemize}

In Figure~\ref{fig:performance} we plot the frame loss ratio for state-dependent symbol-wise DF, Message-wise DF and for IF. We further plot the upper bound for State-dependent symbol-wise DF derived in Section~\ref{sec:upperBound} with $R=2/3$, $T\leq 9$ and $N_1=N_2=N_3=2$ while assuming $\forall j:\alpha_j=\alpha$.

Since enforcing rate $2/3$ resulted with low overall delay ($T=6$ for example in case of message-wise DF), we further plot the performance of all schemes where we force $T=9$ and allow the rate to be greater than or equal to $2/3$ (while trying to find the smaller rate possible). Figure~\ref{fig:performance2} depicts all schemes when the rates are as close as possible to $2/3$ (from above) with $T=9$. We note that, again, state-dependent symbol-wise DF outperforms all other methods.

\begin{figure}[h]
    \centering
    \includegraphics[width=0.85\columnwidth]{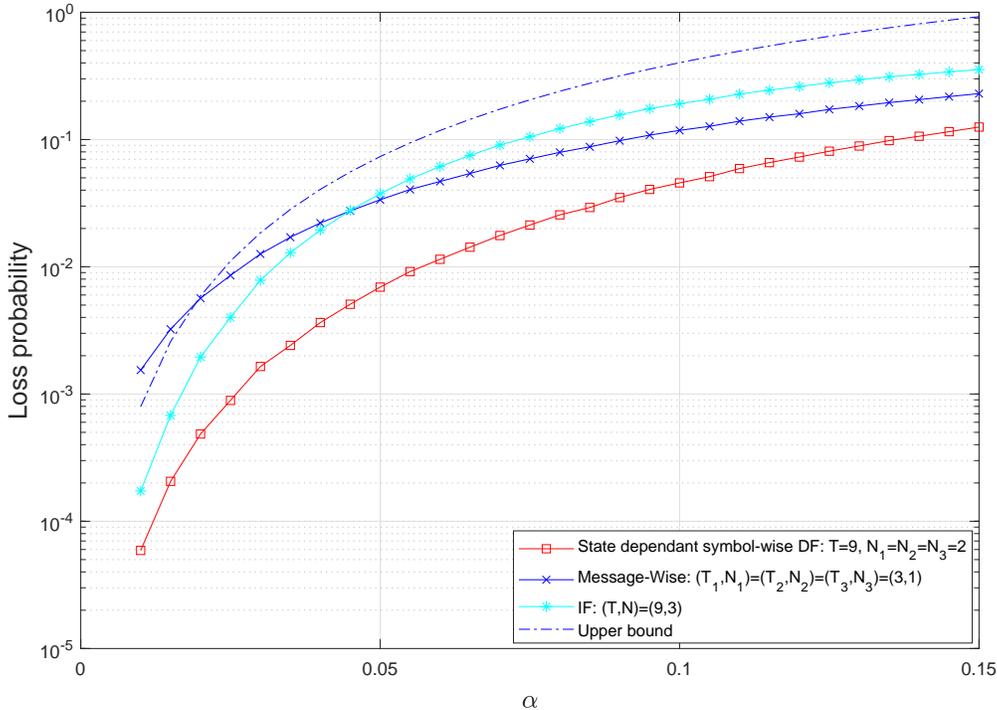}
    \caption{Four-node (two relays) network loss probability for state-dependent symbol-wise DF, message-wise DF and IF with $T \leq 9$, rate $2/3$ and largest $N_1 +N_2+N_3$ where $\alpha$ denotes the erasure probability (same over all hops).}
    \label{fig:performance}
\end{figure}

\begin{figure}[h]
    \centering
    \includegraphics[width=0.85\columnwidth]{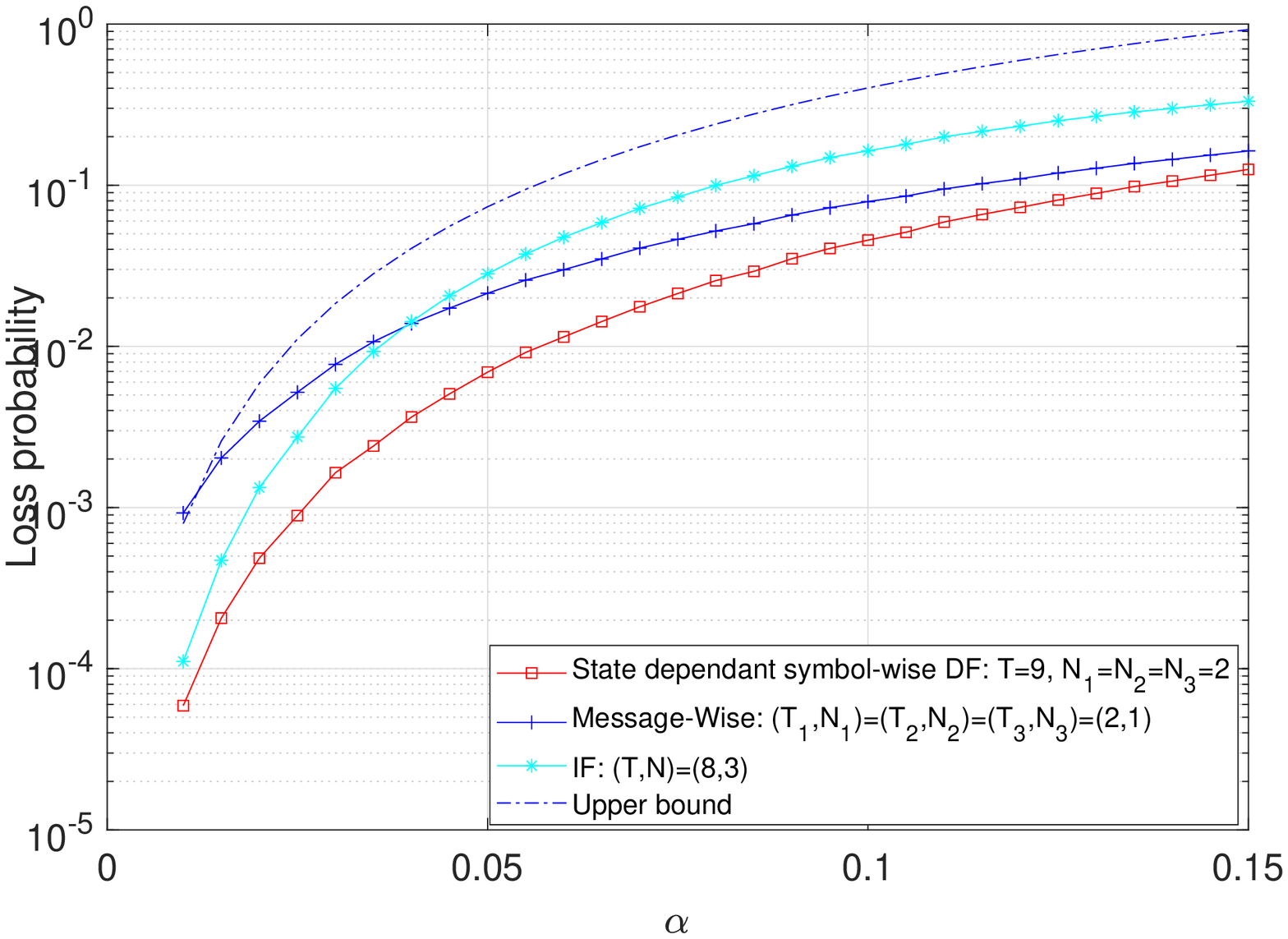}
    \caption{Four-node (two relays) network loss probability for state-dependent symbol-wise DF, message-wise DF and IF with delay $T = 9$, $R\geq2/3$ and largest $N_1 +N_2+N_3$ where $\alpha$ denotes the erasure probability (same over all hops).}
    \label{fig:performance2}
\end{figure}

\section{Extension to Sliding Window Model}
\label{sec:extToSliding}
Consider the following sliding window model. For each $j\in\{1,\ldots,L+1\}$, channel $(r_{j-1},r_j)$ introduces at most $N_j$ erasures in any period of $T+1$ consecutive time slots (sliding window of size $T+1$).

Under the sliding window described above, denote with $C^{+,\rm{sw}}_{T,N_1,\ldots,N_{L+1}}$ as the upper bound on the achievable rate for $(N_1,N_2,\ldots,N_{L+1})$ channel. We further denote $R^{\rm{sw}}$ as the achievable rate under the sliding window model.

Our goal is to show that
\begin{align}
    C^{+,\rm{sw}}_{T,N_1,\ldots,N_{L+1}}\leq C^{+}_{T,N_1,\ldots,N_{L+1}}
    \label{eq:upperEqSW}
\end{align}
and
\begin{align}
R^{\rm{sw}}=R.
\label{eq:achivEqSW}
\end{align}

With respect to the upper bound, since for any $j\in\{1,\ldots,L+1\}$, $N_j$-erasure sequence can be introduced by channel $(r_{j-1},r_j)$ in the sliding window model \eqref{eq:upperEqSW} holds.

Next, we show that the state-dependent scheme can achieve the same rate under the sliding window model. As was shown in Section~\ref{sec:codScheme}, for any $j\in\{1,\ldots,L+1\}$, each symbol can be recovered as long as channel $(r_{j-1},r_j)$ introduces at most $N_j$ erasures in a window of size $n_j$. From \eqref{eq:generalRelay} we have $n_j<T+1$, it follows that all $L$ conditions hold under the sliding window model thus the state-dependent code can recover all symbols and hence \eqref{eq:achivEqSW} holds.

\section{Concluding Remarks}
Streaming codes became an integral part of real-time interactive video streaming applications. The ability to improve the quality of user-experience while meeting stringent latency constraints helped to transform these applications from a niche to one of the fastest-growing segments of IP traffic.

The problem of quantifying the capacity of streaming codes was first addressed for point-to-point links. Later it was extended for three-node networks and in this work, we extended the upper bound of streaming codes to multi-hop relay networks (i.e., for networks with any number of nodes) and further, we suggested an achievable scheme which achieves this upper bound with a gap that vanishes as the field size used gets larger. While for three-node networks, capacity can be achieved with ``state-independent'' symbol-wise DF coding scheme, we showed that this scheme could not be easily extended to more than three-nodes and suggested a symbol-wise DF ``state-dependent'' coding scheme.

While starting by describing the upper bound and proving the achievable scheme over ``deterministic arbitrary erasure channel'',  Section~\ref{sec:extToSliding} outlined that these results also hold for a sliding window model. While this model is a simplified model for a random erasure model, Section~\ref{sec:upperBound} showed that the error probability of the achievable scheme could be upper bounded for random channel and Section~\ref{sec:numerical} provided numerical results showing that the proposed coding scheme outperforms other simple options such as message-wise DF and instantaneous forwarding.

Packet erasures can occur either in sparse patterns or in a bursty manner. The first works on streaming codes analyzed the case of burst erasure channel. Later, works studied channels which could have either a burst or sparse erasures in a given window. Future work may explore streaming codes in a setup with relays in this case. Another avenue is to study streaming codes for more complicated networks.
\appendices
\section{Derivations In The Converse Proof Of Theorem\ref{thm:conv}}
\label{sec:app1}
{\emph Deviation of \eqref{eq:rateOfFirst}}: Since the $(N_1,\ldots,N_{L+1})$-achievable $(n_1,\ldots,n_{L+1},k, T)_{\mathbb{F}}$-streaming code restricted to channel $(r_0,r_1)$ can be viewed as a point-to-point streaming code with rate $k/n_1$ and delay $T-N_2^{L+1}$ which can correct the periodic erasure sequence $\tilde{e}^{\infty}$ illustrated in Figure~\ref{fig:Periodic1}, it follows from the arguments in \cite{badr2017layered} Section IV-A that \eqref{eq:rateOfFirst} holds. For the sake of completeness, we present a rigorous proof below.

Using~\eqref{eq:outcomeFirst}, we have
\begin{align}
    |\mathbb{F}|^{k\times j(T-N_2^{L+1}+1)}\leq |\mathbb{F}|^{n\times (j+1)(T-N_1^{L+1}+1)}
    \label{eq:App1FirstSeg_1}
\end{align}
because $j(T-N_2^{L+1}+1)$ source messages can take $|\mathbb{F}|^{k\times j(T-N_2^{L+1}+1)}$ values and $(j+1)(T-N_1^{L+1}+1)$ source packets can take at most $|\mathbb{F}|^{n\times (j+1)(T-N_1^{L+1}+1)}$ values for each $j$. Taking logarithm on both sides of \eqref{eq:App1FirstSeg_1} followed by dividing both sides by $j$, we have
\begin{align}
    k(T-N_2^{L+1}+1) \leq n(1 + 1/j)(T-N_1^{L+1}+1)
    \label{eq:App1FirstSeg_2}
\end{align}
Since \eqref{eq:App1FirstSeg_2} holds for all $j\in\mathbb{N}$, it follows that \eqref{eq:rateOfFirst} hold.

{\emph Deviation of \eqref{eq:rateOfLast}}: Since the $(N_1,\ldots,N_{L+1})$-achievable $(n_1,\ldots,n_{L+1},k, T)_{\mathbb{F}}$-streaming code restricted to channel $(r_{L},r_{L+1})$ can be viewed as a point-to-point streaming code with rate $k/n_{L+1}$ and delay $T-N_1^{L}$ which can correct the periodic erasure sequence $\hat{e}^{\infty}$ illustrated in Figure~\ref{fig:Periodic2}, it follows from the arguments in \cite{badr2017layered} Section IV-A that \eqref{eq:rateOfLast} holds. For the sake of completeness, we present a rigorous proof below.

Using~\eqref{eq:outcomeLast}, we have
\begin{align}
    |\mathbb{F}|^{k\times j(T-N_1^{L}+1)}\leq |\mathbb{F}|^{n\times (j+1)(T-N_1^{L+1}+1)}
    \label{eq:App1LastSeg_1}
\end{align}
because $j(T-N_1^{L}+1)$ source messages can take $|\mathbb{F}|^{k\times j(T-N_1^{L}+1)}$ values and $(j+1)(T-N_1^{L+1}+1)$ source packets can take at most $|\mathbb{F}|^{n\times (j+1)(T-N_1^{L+1}+1)}$ values for each $j$. Taking logarithm on both sides of \eqref{eq:App1LastSeg_1} followed by dividing both sides by $j$, we have
\begin{align}
    k(T-N_1^{L}+1) \leq n(1 + 1/j)(T-N_1^{L+1}+1)
    \label{eq:App1LastSeg_2}
\end{align}
Since \eqref{eq:App1LastSeg_2} holds for all $j\in\mathbb{N}$, it follows that \eqref{eq:rateOfLast} hold.

{\emph Deviation of \eqref{eq:outcomeMid}}: Since the $(N_1,\ldots,N_{L+1})$-achievable $(n_1,\ldots,n_{L+1},k, T)_{\mathbb{F}}$-streaming code restricted to channel $(r_{j-1},r_{j})$ can be viewed as a point-to-point streaming code with rate $k/n_{j}$ and delay $T-N_1^{j-1}-N_{j+1}^{L+1}$ which can correct the periodic erasure sequence $\breve{e}^{\infty}$ illustrated in Figure~\ref{fig:Periodic3}, it follows from the arguments in \cite{badr2017layered} Section IV-A that \eqref{eq:generalRelay} holds. For the sake of completeness, we present a rigorous proof below.

Using~\eqref{eq:outcomeMid}, we have
\begin{align}
    |\mathbb{F}|^{k\times j(T-N_1^{j-1}-N_{j+1}^{L+1}+1)}\leq |\mathbb{F}|^{n\times (j+1)(T-N_1^{L+1}+1)}
    \label{eq:App1MidSeg_1}
\end{align}
because $j(T-N_1^{L}+1)$ source messages can take $|\mathbb{F}|^{k\times j(T-N_1^{j-1}-N_{j+1}^{L+1}+1)}$ values and $(j+1)(T-N_1^{j-1}-N_{j+1}^{L+1}+1)$ source packets can take at most $|\mathbb{F}|^{n\times (j+1)(T-N_1^{L+1}+1)}$ values for each $j$. Taking logarithm on both sides of \eqref{eq:App1LastSeg_1} followed by dividing both sides by $j$, we have
\begin{align}
    k(T-N_1^{j-1}-N_{j+1}^{L+1}+1) \leq n(1 + 1/j)(T-N_1^{L+1}+1)
    \label{eq:App1MidSeg_2}
\end{align}
Since \eqref{eq:App1MidSeg_2} holds for all $j\in\mathbb{N}$, it follows that \eqref{eq:rateOfLast} hold.

\section{Examples for rate change in relay}
\label{app:rateChangeInRelay}
As mentioned above, relay $r_j$ may need to increase or decrease the rate of the code used by relay $r_{j-1}$. Below, we show examples for the following two cases:
\begin{itemize}
    \item $\frac{k}{n_{j+2}}>\frac{k}{n_{j+1}}$. This means that $n_{j+2}<n_{j+1}$, i.e., that the block size of the MDS code used by relay $r_{j+1}$ is smaller than the block size used by relay $r_{j}$. At time $i+T-\sum_{j+1}^{L+1}N_l+1$, node $r_{j+1}$ can recover the original data and send any of the erased symbols of the code used by $r_{j}$.

    An example is given in Table~\ref{tab:exampleForIncreaseRate} for $N_{j+1}=2,~N_{j+2}=1,~T'=4$ (where $T'=T-\sum_{l=1,l\neq j+1,j+2}N_l$).
    We note that in this example $k=T-\sum N_l+1=T'-N_{j+1}-N_{j+2}=2$.

    Relay $r_{j+1}$ forwards the first $k-1=1$ symbols it receives. At $i+N_{1}^j+3$ the relay can recover the original data, hence from this point it sends (for example) the erased symbols.

\begin{table}[h]


\begin{center}
\begin{tabular} {|c|c|c|c|c|}
     \hline
     $i+N_{1}^{j}$ & $i+N_{1}^{j}+1$& $i+N_{1}^{j}+2$ & $i+N_{1}^{j}+3$ & $i+N_{1}^{j}+4$  \\ \hline
    \multicolumn{5}{|c|}{Link $(r_j,r_{j+1})$} \\ \hline
     \tikzmark{topD}{\textcolor{black}{a_{i}}} & & \tikzmark{topE}{} & & \\ \hline
     & $\textcolor{black}{b_{i+1}}$ & & & \\ \hline
     & &  \cellcolor{gray!25}{$\textcolor{black}{f^1(a_i,b_{i+1})}$} & & \\ \hline
     \tikzmark{rightD}{} & & \tikzmark{rightE}{} & \cellcolor{gray!25}{$\textcolor{black}{f^2(a_i,b_{i+1})}$} & \\ \hline
     \multicolumn{5}{|c|}{Link $(r_{j+1},r_{j+2})$} \\ \hline
     & & $\textcolor{black}{b_{i+1}}$ & &   \\ \hline
     & & & $\textcolor{black}{a_{i}}$ &  \\ \hline
     & & & & \cellcolor{gray!25}{$\textcolor{black}{f^1(a_i,b_{i+1})}$}  \\ \hline
\end{tabular}
\DrawVLine[black, thick, opacity=0.8]{topD}{rightD}
\DrawVLine[black, thick, opacity=0.8]{topE}{rightE}
\end{center}
\caption{Example of increasing the rate between links. In this example, $N_{j+1}=2,~N_{j+2}=1,~T'=4$, hence $\frac{k}{n_{j+2}}=2/4<2/3=\frac{k}{n_{j+1}}$. Assuming symbol $i+N_{1}^j$ and $i+N_1^j+2$ were erased when transmitted in link $(r_{j},r_{j+1})$. Parity symbol are shaded.}
\label{tab:exampleForIncreaseRate}
\end{table}
    \item $\frac{k}{n_{j+2}}<\frac{k}{n_{j+1}}$. This means that $n_{j+2}>n_{j+1}$, i.e., the block size of the code used by relay $r_{j+1}$ is larger than the block size used by relay $r_{j}$.

    At time $i+T-\sum_{j+1}^{L+1}N_l+1$, relay $r_{j+1}$ can again recover the original data and hence transmit additional $n_{j+2}-k$ symbols needed to allow handling any $N_{j+2}$ erasures in the link $(r_{j+1},r_{j+2})$.

    An example is given in Table~\ref{tab:exampleForReduceRate} for $N_{j+1}=1,~N_{j+2}=2,~T'=4$ (where, again, $T'=T-\sum_{l=1,l\neq j+1,j+2}$).

    Relay $r_{j+1}$ forwards the first $k-1=1$ symbols it receives. At $i+N_{1}^j+2$, the relay can recover the original data, hence from this point is sends (for example) the erased symbols while adding parity symbols to reach the required rate.

    \begin{table}[h]
    \begin{center}
    \begin{tabular} {|c|c|c|c|c|}
         \hline
         $i+N_{1}^{j}$ & $i+N_{1}^{j}+1$& $i+N_{1}^{j}+2$ & $i+N_{1}^{j}+3$ & $i+N_{1}^{j}+4$  \\ \hline
        \multicolumn{5}{|c|}{Link $(r_j,r_{j+1})$} \\ \hline
        \tikzmark{topD}{\textcolor{black}{a_{i}}} & & & & \\ \hline
         & $\textcolor{black}{b_{i+1}}$ & & & \\ \hline
         \tikzmark{rightD}{} & &  \cellcolor{gray!25}{$\textcolor{black}{f^1(a_i,b_{i+1})}$} & & \\ \hline
         \multicolumn{5}{|c|}{Link $(r_{j+1},r_{j+2})$} \\ \hline
         & $\textcolor{black}{b_{i+1}}$ & & &  \\ \hline
         & & $\textcolor{black}{a_i}$ & &  \\ \hline
         & & &  \cellcolor{gray!25}{$\textcolor{black}{f^1(a_i,b_{i+1})}$} &  \\ \hline
         & & & & \cellcolor{gray!25}{$\textcolor{black}{f^2(a_i,b_{i+1})}$}  \\ \hline
    \end{tabular}
    \DrawVLine[black, thick, opacity=0.8]{topD}{rightD}
    \end{center}
    \caption{Example of reducing rate between nodes. In this example, $N_{j+1}=1,~N_{j+2}=2,~T'=4$, hence $\frac{k}{n_{j+1}}=2/3>2/4=\frac{k}{n_{j+2}}$. Assuming symbol $i+N_1^j$ was erased when transmitted in link $(r_{j},r_{j+1})$.Parity symbols are shaded.}
    \label{tab:exampleForReduceRate}
    \end{table}
\end{itemize}

\bibliographystyle{IEEEtran}
\bibliography{eladd.bib}

\end{document}